\documentclass[utf8]{FrontiersinHarvard}

\usepackage{bm}
\usepackage{mathtools}
\usepackage[loadonly]{enumitem}

\usepackage{algorithm}
\usepackage{algorithmic}

\usepackage{url,hyperref,lineno,microtype,subcaption}
\usepackage[onehalfspacing]{setspace}

\usepackage{jhoagg}

\usepackage[capitalize,noabbrev]{cleveref}

\newtheorem{example}{Example}
\newtheorem{assumption}{Assumption}

\newtheorem{proposition}{Proposition}
\newtheorem{theorem}{Theorem}
\usepackage{graphicx}
\usepackage{subcaption}

\newcounter{enumA}
\newcounter{enumO}

\newlist{enumA}{enumerate}{1}
\setlist[enumA,1]{%
  label=(A\arabic*),
  leftmargin=1cm,
  before=\setcounter{enumAi}{\value{enumA}},
  after=\setcounter{enumA}{\value{enumAi}}
}

\newlist{enumO}{enumerate}{1}
\setlist[enumO,1]{%
  label=(O\arabic*),
  leftmargin=1cm,
  before=\setcounter{enumOi}{\value{enumO}},
  after=\setcounter{enumO}{\value{enumOi}}
}

\crefname{equation}{}{}

\newcommand{\SOT}{\mathrm{SO}(3)}
\newcommand{\SET}{\mathrm{SE}(3)}
\newcommand{\ubar}[1]{{\underline{#1}}}

\newcounter{appsection}
\renewcommand{\theappsection}{\Alph{appsection}}

\crefname{appsection}{Appendix}{Appendices}
\Crefname{appsection}{Appendix}{Appendices}

\newcommand{\appsection}[2][]{%
    \refstepcounter{appsection}%
    \section*{\texorpdfstring{APPENDIX \theappsection \\ #2}{Appendix #2}}%
    \addcontentsline{toc}{section}{Appendix \theappsection: #2}%
    \ifx&#1&\else\label{#1}\fi
}

\def\keyFont{\fontsize{8}{11}\helveticabold }

\def\firstAuthorLast{Felipe Arenas-Uribe {et~al.}} 

\def\Authors{First Author\,$^{1,*}$, Co-Author\,$^{2}$ and Co-Author\,$^{1,2}$}

\def\Authors{Felipe Arenas-Uribe\,$^{1,*}$, T. Michael Seigler\,$^{1}$ and Jesse B. Hoagg\,$^{1}$}

\def\Address{$^{1}$Department of Mechanical and Aerospace Engineering, University of Kentucky, Lexington, KY, USA}

\def\corrEmail{felipearur@uky.edu}

\begin{document}
\onecolumn
\firstpage{1}

\title[Safe Landing on Small Celestial Bodies]{Safe Position and Attitude Control for Landing on Small Celestial Bodies with Gravitational Uncertainty} 

\author[\firstAuthorLast ]{\Authors}
\address{}
\correspondance{}

\extraAuth{}

\funding{This work is supported in part by the National Science Foundation (1932105, 2450718) and USDA NIFA (2024-69014-42393).}

\maketitle

\begin{abstract}
Landing on small celestial bodies is challenging because the spacecraft must satisfy safety and actuator constraints despite gravitational forces that are difficult to model accurately.
This article presents an optimal control for safe landing with gravitational uncertainty and strict actuator limits. 
The approach addresses spacecraft pose control with attitude represented on $\SOT$, where there are actuator constraints and state constraints on position, attitude, velocity, and
angular velocity.
The approach combines several key techniques. 
First, the gravitational force uncertainty is estimated using an extended high-gain observer, and we present a new dynamic upper bound on the estimation error. 
This gravitational estimate and dynamic bound are then used to compute optimal control forces and torques that satisfy state and actuator constraints while tracking a landing trajectory. 
Optimal forces and torques are obtained from the closed-form solution to a quadratic program that has a single control barrier function constraint constructed by composing multiple control barrier functions that are designed to enforce each state and actuator constraint.
Finally, an optimal allocation maps control forces and torques to actuator commands that satisfy actuator constraints.
The method is demonstrated in simulation using 2 vehicle configurations with  severe gravitational uncertainty.

\tiny
 \keyFont{\section{Keywords:} Spacecraft control, constrained control, optimal control, control barrier functions, polytope projection} 
\end{abstract}

\section{Introduction}

Small celestial bodies (SCBs) have attracted attention for their scientific value and relevance to planetary defense. 
Sample return missions provide insights into the formation and composition of celestial bodies. 
Example missions include NASA OSIRIS-REx, JAXA Hayabusa, and ESA Rosetta. 
In addition, near-Earth objects motivate active SCB monitoring, and the rapidly developing space economy has generated interest in space mining~\cite[]{modell_space_2025}.

Soft landing on SCBs is a proximity maneuver that requires precise control of position and attitude, while satisfying actuator limits and safety constraints~\cite[]{zhang_coupled_2023}. 
Spacecraft actuator systems are typically overactuated to provide redundancy as a passive fault-tolerant strategy against actuator failures \cite[]{hasan_fault-tolerant_2022}. 
Nevertheless, soft landing is a challenge because of the uncertain exogenous force due to disturbances such as solar radiation pressure and unmodeled gravitational effects~\cite[]{mammarella_controls_2025}. 
Error in the mass estimate for a SCB is typically in the range of 10--15\%~\cite[]{bull_optical_2021}. 
In addition, SCBs have irregular shapes and nonuniform mass distribution. 
To account for these effects, higher-order models can be used~\cite[]{scheeres_orbital_2012}. 
However, accurately determining model parameters requires high-resolution remote sensing or in-situ measurements that are expensive to acquire~\cite[]{penarroya_orbit_2019}. 
Even in the best case, the gravitational field model has residual error, which needs to be addressed by the control system.

Spacecraft guidance and control often uses a modular architecture, where the guidance module generates a reference trajectory and the feedback controller is responsible for following the reference. 
For a review of trajectory generation for SCB landing, see \cite[Section 3.3]{malyuta_advances_2021}. 
Spacecraft pose control is addressed using a variety of representations, including position-attitude on $\BBR^3 \times \SOT$, dual quaternions, and the special Euclidean group $\SET$~\cite[]{zhang_coupled_2023}. 
Spacecraft pose control methods include integrated, adaptive, and disturbance rejection approaches \cite[]{wang_integrated_2019,sun_adaptive_2018,zhang_adaptive_2022,zhang_coupled_2023}.

For position-only control, robust approaches such as $H_\infty$ \cite[]{simplicio_synthesis_2018}, sliding mode \cite[]{fu_orbitattitude_2022}, and adaptive \cite[]{tiwari_direct_2023,sun_adaptive_2018} have been developed. 
Linear extended high-gain observers (EHGO) have also been applied to estimate and compensate for gravitational uncertainties \cite[]{dunham_constrained_2016}. Although these approaches can improve robustness and reference tracking performance, they generally do not provide guarantees of safety constraint satisfaction. 
On the other hand, model-predictive control (MPC) can explicitly incorporate state and input constraints.
MPC methods have been developed for spacecraft proximity operations using classical and dual-quaternion-based formulations \cite[]{tiwari_direct-adaptive_2022,van_leeuwen_nonlinear_2022,sun_dual_2021}. 
However, these methods typically rely on approximations and discretizations of the nonlinear spacecraft dynamics.
Consequently, guaranteeing recursive feasibility and constraint satisfaction often requires constraint-tightening and robustness mechanisms, which can significantly increase computational complexity.

These limitations motivate the use of control barrier functions (CBFs), which provide a framework to enforce state constraints through forward invariance \cite[]{Wieland_CBF_2007,ames_control_2019}. 
CBF methods for spacecraft have been developed for orbit transfers~\cite[]{ong_intermittent_2023}, proximity operations~\cite[]{breeden_robust_2023}, and formation flying~\cite[]{kamat_electromagnetic_2026}. 
However, uncertainty in the gravitational field makes SCB landing a challenging application for CBF methods. 
Although robust CBF approaches \cite[]{jankovic_robust_2018} address model uncertainty, synthesizing a valid CBF subject to actuator constraints is difficult, and model uncertainty only exacerbates this challenge. 
An alternative approach is to estimate the model uncertainty using a disturbance observer~\cite[]{das_robust_2022,alan_disturbance_2023,sun_safety-critical_2024,alan_generalizing_2025,wang_safety-critical_2025,das_robust_2025}. 
However, these approaches rely on bounds on disturbance derivative to characterize observer error, and more importantly, they do not address the challenge of actuator constraints.

This article presents a new control method for safe landing on a SCB with gravitational uncertainty.
The approach addresses spacecraft pose control with attitude represented on $\SOT$, where there are strict state constraints on position (i.e., glideslope constraint), attitude (i.e., boresight constraint), velocity, and angular velocity, as well as strict actuator constraints. 
This safe landing approach relies on several key building blocks.  
First, \cref{sec:Constraint_Projection} uses a polytope projection to map the set of admissible actuator inputs into a lower-dimensional set of admissible forces and torques. 
This process transforms the control variable from actuator inputs to control forces and torques, and transforms the actuator constraints into equivalent constraints on the control force and torque. 
This process is analogous to \cite{durham_attainable_1994}.
This section then presents an optimal control allocation to map admissible forces and torques back to admissible actuator commands. 
Next, \cref{sec:FeedbackControl} presents a nonlinear EHGO to estimate gravitational uncertainty. 
Notably, we develop a dynamic upper bound on the estimation error, which is critical for guaranteeing safety. 
This section also uses the gravitational force estimate to develop desired control forces and torques for tracking the reference trajectory, which is assumed to be generated and periodically updated by a guidance module.

\cref{sec:CBFControl} presents the main contribution of this article, which is an optimal feedback control for safe landing that satisfies all state and actuator constraints. 
Although the desired control forces and torques achieve landing, they do not generally satisfy state and actuator constraints. 
Thus, we present a method inspired from \cite[]{rabiee_closed-form_2024,rabiee_composition_2024} to construct a composite CBF  that can be used to ensure all state and actuator constraints are satisfied. 
The approach in \cite[]{rabiee_closed-form_2024,rabiee_composition_2024} does not address model uncertainty, so we extend the method to address gravitational uncertainty by using the dynamic upper bound on the estimation error (from \cref{sec:FeedbackControl}). 
Note that the dynamic bound is less conservative than a standard robust CBF approach~\cite[]{das_robust_2025}. 
Then, the desired forces and torques and the composite CBF are used to construct a closed-form optimal control that achieves landing while satisfying state and actuator constraints. 
This safe and optimal control is the solution to a quadratic program (QP) that can be expressed in closed form, which is computationally efficient for on-spacecraft implementation. 

Finally, \cref{sec:Simulations} presents simulation studies using 2 vehicle configurations: a robotic lander similar to the Hayabusa2 mission \cite[]{ho_mascot_2021}, and a crewed landing module similar to the Apollo lunar module \cite[]{nasa1969lunarmodule}. 
The results demonstrate safe landing under severe gravitational uncertainty where traditional feedback control strategies fail. 
Some preliminary results related to this article appeared in the conference article \cite[]{arenas-uribe_safe_2025}. 
However, this article goes beyond the preliminary work by addressing position-and-attitude control rather than position-only control.

\section{Notation}

Physical vectors are denoted with bold symbols (e.g., $\bm{r}$). 
A frame is a collection of mutually orthogonal physical unit vectors. 
If $\SF$ is a frame, then $[\bm{r}]_{\SF}$ is $\bm{r}$ resolved in $\SF$. 
The time derivative of $\bm{r}$ with respect to the frame $\SF$ is denoted by $^{\SF}\dot{\bm{r}}$.

The special orthogonal group $\SOT$ is the set of orthogonal matrices in $\BBR^{3 \times 3}$ with determinant one. 
The geodesic distance between $P \in \SOT$ and $Q \in \SOT$ is 
\begin{equation*}
    d(P,Q) \triangleq \mathrm{arccos} \frac{\mathrm{tr}\, P^{\top}Q-1}{2} \in [0,\pi],
\end{equation*}
where $\mathrm{tr}$ is the trace. 
Note that $d(P,Q)=0$ if and only if $P=Q$.

The set of skew-symmetric matrices in $\BBR^{3 \times 3}$ is denoted by $\mbox{so}(3)$. 
Let $\omega \in \mathbb{R}^3$.
Then, $[\omega]_{\times} \in \mbox{so}(3)$ is the skew-symmetric representation of $\omega$, and $(\cdot)^{\vee}$ is the inverse of $[\cdot]_{\times}$, that is, $([\omega]_{\times})^{\vee} = \omega$.

Let $h \colon \BBR^n \rightarrow \BBR^l$ be continuously differentiable. 
Then, $h^\prime \colon \BBR^n \rightarrow \BBR^{l \times n}$ is defined as $h'(x) \triangleq \frac{\partial h (x)}{\partial x}$. 
The Lie derivatives of $h$ along the vector field of $\psi:\BBR^n \rightarrow \BBR^{n \times l}$ is $L_{\psi}h \triangleq h'(x) \psi(x)$. 
Let $\| \cdot \|$ denote the 2-norm. The boundary of $\mathcal{A} \subseteq \BBR^n$ is denoted by $\text{bd }\mathcal{A}$.

\section{Problem Formulation}\label{sec:ProblemFormulation}

\subsection{Equations of Motion}

Consider a SCB (referred to as an asteroid) and a nearby spacecraft (see \Cref{fig:Asteroid_Spacecraft}). 
We assume that the asteroid's mass center is an inertial reference point satisfying Newton's second law.
Let $\SF_{\rm I}$ be an inertial frame. 
Frame $\SF_{\rm A}$ is fixed to the asteroid, and frame $\SF_{\rm B}$ is fixed to the spacecraft. 
The angular velocity of $\SF_{\rmA}$ relative to $\SF_{\rmI}$ is $\bm{\omega}_{\rm A}$. 
We assume that there are no external moments acting on the asteroid (i.e., $^{\SF_{\rm I}}\dot{\bm{\omega}}_{\rm A} = \bm{0}$), and $\SF_{\rm I}$ is oriented such that $[\bm \omega_{\rm A}]_{\SF_\rmA} = [\bm \omega_{\rm A}]_{\SF_\rmI}$. 
Let $\bm{\omega}_{\rm B}$ be the angular velocity of $\SF_{\rm B}$ relative to $\SF_{\rm I}$, and let $\bm{\omega}_{\rm{B/A}}$ be the angular velocity of $\SF_{\rmB}$ relative to $\SF_{\rmA}$.

\begin{figure}[ht]
    \centering
    \includegraphics[width=0.35\columnwidth]{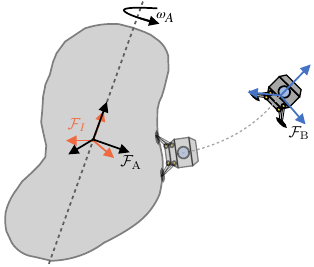}
    \caption{Spacecraft, asteroid, and 3 frames.}
    \label{fig:Asteroid_Spacecraft}
\end{figure}

Let $\omega \triangleq [\bm{\omega}_{\rm B/A}]_{\SF_{\rm B}}$ and $R:[0,\infty) \to \SOT$ be the rotation matrix of $\SF_{\rm B}$ relative to $\SF_{\rm A}$, which satisfies
\begin{equation} \label{eq:Attitude_dyn}
    \dot{R}(t) = R(t)[\omega(t)]_{\times},
\end{equation}
where $R(0) = R_{0}$.

Let $\bm{J}_{\rmB}$ be the spacecraft's physical inertia matrix relative to its mass center, and let $\bm{T}$ be the resultant torque acting on the spacecraft due to the spacecraft actuators. 
The spacecraft's angular momentum relative to its mass center is
\begin{equation*}
    \bm{H} \triangleq \bm{J}_{\rm B}\bm{\omega}_{\rm B}.
\end{equation*}
Since $\bm{\omega}_{\rm B} = \bm{\omega}_{\rm A} + \bm{\omega}_{\mathrm{B/A}}$ and $^{\SF_{\rm I}}\dot{\bm{\omega}}_{\rm A} = \bm{0}$, the time derivative of $\bm H$ with respect to $\SF_{\rm I}$ is
\begin{equation*}
    ^{\SF_{\rm I}} \dot{\bm{H}} = \bm{J}_{\rm B} (\, ^{\SF_{\rm B}} \dot{\bm{\omega}}_{\mathrm{B/A}} + (\bm{\omega}_{\rm A} + \bm{\omega}_{\mathrm{B/A}}) \times \bm{\omega}_{\mathrm{B/A}}) + (\bm{\omega}_{\rm A} + \bm{\omega}_{\mathrm{B/A}}) \times \bm{J}_{\rm B} (\bm{\omega}_{\rm A} + \bm{\omega}_{\mathrm{B/A}}).
\end{equation*}
Thus, Euler's rotation law implies
\begin{equation*}
    \bm{J}_{\rm B} (\, ^{\SF_{\rm B}} \dot{\bm{\omega}}_{\mathrm{B/A}} + (\bm{\omega}_{\rm A} + \bm{\omega}_{\mathrm{B/A}}) \times \bm{\omega}_{\mathrm{B/A}}) + (\bm{\omega}_{\rm A} + \bm{\omega}_{\mathrm{B/A}}) \times \bm{J}_{\rm B} (\bm{\omega}_{\rm A} + \bm{\omega}_{\mathrm{B/A}}) = \bm{T},
\end{equation*}
which is resolved in $\SF_{\rm B}$ as
\begin{equation} \label{eq:angvel_dyn}
    J_{\rm B}\dot{\omega}(t) + J_{\rm B}[R(t)^\top \omega_{\rm A} + \omega(t)]_{\times} \omega(t) + [R(t)^\top \omega_{\rm A} + \omega(t)]_{\times}J_{\rm B} \big(R(t)^\top \omega_{\rm A} + \omega(t)\big) = T(t),
\end{equation}
where $\omega_{\rm A} \triangleq [\bm{\omega}_{\rm A}]_{\SF_{\rm A}}$, $J_{\rmB} \triangleq [\bm{J}_{\rmB}]_{\SF_{\rmB}}$, $T \triangleq [\bm{T}]_{\SF_{\rm B}}$, and $\omega(0)=\omega_{0}$.

The position and velocity of the spacecraft's mass center relative to the asteroid's mass center are $\bm{r}$ and $\bm{v} \triangleq {}^{\SF_{\rm I}}\dot{\bm{r}}$. 
The spacecraft mass $m>0$ is assumed constant. 
Define $r \triangleq [\bm{r}]_{\SF_{\rm A}}$, and let $\bm{G}(\bm{r})$ be the mass-normalized asteroid's gravitational force on the spacecraft. 
Let $\bm{F}$ be the resultant force acting on the spacecraft due to the spacecraft actuators. 
Thus, Newton's second law implies
\begin{equation*}
    m \, ^{\SF_{\rm I}} \dot{\bm{v}} = m\bm{G}(\bm{r}) + \bm{F},
\end{equation*}
which is resolved in $\SF_{\rm A}$ as
\begin{align}
\dot{r}(t) &= v(t), \label{eq:pos_dyn} \\
\dot{v}(t) &= -2[\omega_{\rm A}]_{\times}v(t) -[\omega_{\rm A}]_{\times}^2r(t) + G(r) + \frac{1}{m}R(t)F(t), \label{eq:vel_dyn}
\end{align}
where $v \triangleq [\bm{v}]_{\SF_{\rm A}}$, $[\bm{G}(\bm{r})]_{\SF_{\rm A}} =  G(r)$, $F \triangleq [\bm{F}]_{\SF_{\rm B}}$, $r(0) = r_0$, and $v(0) = v_0$.

The gravitational force $G$ is uncertain because of  uncertainty in mass and mass distribution.
Let $G_{\rm m}: \mathbb{R}^3 \to \mathbb{R}$ be a known gravitational model, which is 2-times continuously differentiable. 
Then, the gravitational model error is
\begin{equation*}
    \Delta(r) \triangleq G(r) - G_{\rm m}(r).
\end{equation*}
We make the following technical assumption. 

\begin{assumption}\label{ass:Bounded_Dist}
\rm
$\Delta$ and $\Delta^\prime$ are bounded. 
\end{assumption}
The Newtonian point-mass model is one possible choice for $G_{\rm m}$. 
In cases where more information about the asteroid geometry and composition is available, spherical harmonic models can be a more accurate representation \cite[]{scheeres_orbital_2012}.

\subsection{Asteroid Soft-Landing Problem}

The equations of motion  \cref{eq:pos_dyn,eq:vel_dyn,eq:Attitude_dyn,eq:angvel_dyn} can be written as
\begin{equation} \label{eq:system_dynamics}
    \dot{x}(t) = f(x(t)) + g(x(t))\mu(t) + D w(x(t)),
\end{equation}
where $n=18$ and 
\begin{gather}
x \triangleq \begin{bmatrix} r \\ v \\ \mathrm{vec} \, R \\ \omega \end{bmatrix}, \quad
x_0 \triangleq \begin{bmatrix} r_0 \\ v_0 \\ \mathrm{vec} \, R_0 \\ \omega_0 \end{bmatrix}, \quad 
\mu \triangleq \begin{bmatrix} F \\ T\end{bmatrix}, \label{eq:system_dynamics.2} \\
f(x) \triangleq 
    \begin{bmatrix}
        v\\
        -2[\omega_{\rm A}]_{\times}v -[\omega_{\rm A}]_{\times}^2r + G_{\rm m}(r) \\
        \mathrm{vec}\, R[\omega]_{\times}\\
        \scriptstyle{-[R^\top \omega_{\rm A} + \omega]_{\times} \omega - J_{\rm B}^{-1} \big( [R^\top \omega_{\rm A} + \omega]_{\times}J_{\rm B}(R^\top \omega_{\rm A} + \omega) \big)}
    \end{bmatrix}, \label{eq:system_dynamics.3} \\
g(x) \triangleq 
    \begin{bmatrix}
        0_{3 \times 3} & 0_{3 \times 3}\\
        \frac{1}{m}R & 0_{3 \times 3}\\
        0_{9 \times 3} & 0_{9 \times 3}\\
        0_{3 \times 3} & J_{\rm B}^{-1}
    \end{bmatrix}, \quad 
D \triangleq \begin{bmatrix}
                0_{3 \times 3} \\
                I_{3 \times 3} \\
                0_{9 \times 3} \\
                0_{3 \times 3}
            \end{bmatrix},\quad
w(x) \triangleq \Delta(r).\label{eq:system_dynamics.4}
\end{gather}
The control force $F$ and control torque $T$ depend on $m \geq 6$ actuator inputs (e.g., thrusters) denoted by $u(t) \in \SU$, where the set of admissible actuator inputs $\SU \subseteq \BBR^m$ is a convex polytope and 
\begin{equation}\label{eq:allocation}
    \mu = M u, 
\end{equation}
where $M\in\mathbb{R}^{6\times m}$ and $\mathrm{rank} \, M=6$, which implies that the system is fully actuated or overactuated.

The soft-landing maneuver is a controlled touchdown on the surface of the asteroid, where the spacecraft achieves a final position $r_{\rm f} \in \BBR^3$ and final attitude $R_{\rm f} \in \SOT$ with zero translational and angular velocity--all within predefined tolerances.
Let $\theta_{\rm a} \in (0,\pi]$ be the allowable final attitude error, and let $r_{\rm a},v_{\rm a},\omega_{\rm a}>0$ be the allowable landing radius, final speed, and final angular speed, respectively. 
Then, the set of allowable landing states is 
\begin{equation*}
    L \triangleq \{ x \in \BBR^n: \|r-r_{\rm f}\| \leq r_{\rm a}, \|v\| \leq v_{\rm a},
    d(R, R_{\rm f}) \leq \theta_{\rm a}, \|\omega\| \leq \omega_{\rm a} \}.
\end{equation*}
We say that the solution $x$ to \cref{eq:system_dynamics,eq:system_dynamics.2,eq:system_dynamics.3,eq:system_dynamics.4} \textit{lands} if there exists $t_{\rm l} \in [0, \infty)$ such that $x(t_{\rm l}) \in L$.

In addition to landing, the spacecraft must satisfy state (e.g., safety) and actuator constraints. 
Let $\theta_{\rm gs} \in (0, \frac{\pi}{2}]$ be the landing glideslope angle, and consider the glideslope barrier function 
\begin{equation}
    \psi_1(x) \triangleq \hat{e}_{\rm s}^{\top}(r-\beta r_{\rm f}) 
    - \|r-\beta r_{\rm f}\| \cos\theta_{\rm gs}, 
\end{equation}
where $\hat{e}_s \in \BBR^3$ is the perpendicular unit vector at the landing position $r_\rmf$ resolved in $\SF_{\rm A}$, and $\beta \triangleq 1 - \frac{r_{\rm a}}{\|r_{\rm f}\| \tan \theta_{\rm gs}}$ is a relaxation of the landing position based on the allowable landing radius $r_{\rm a}$. 
The zero-superlevel set of $\psi_1$ is 
\begin{equation*}
    \Psi_1 \triangleq \{ x \in \BBR^n \colon \psi_1(x) \ge 0 \},
\end{equation*}
which is the set of states that satisfy the glideslope constraint.
In this case, the spacecraft is in the allowable approach cone as shown in \cref{fig:gs_constraint}.

Next, let $\theta_{\rm b} \in (0, \frac{\pi}{2}]$ be the boresight angle, and consider the boresight barrier function 
\begin{equation}
\psi_2(x) \triangleq \hat{e}_{\rm b}(R)^{\top}(\beta r_{\rm f}-r) - \| \beta r_{\rm f} - r\| \cos\theta_{\rm b},
\end{equation}
where $\hat{e}_{\rm b}(R) \triangleq -R \begin{bmatrix}0 & 0 & 1\end{bmatrix}^{\top}$ is the boresight direction. 
The zero-superlevel set of $\psi_2$ is 
\begin{equation*}
    \Psi_2 \triangleq \{ x \in \BBR^n \colon \psi_2(x) \ge 0 \},
\end{equation*}
which is the set of states that satisfy the boresight constraint.
In this case, the spacecraft is in the allowable boresight cone, where its attitude is such that it points at the landing position as shown in \cref{fig:bs_constraint}.

\begin{subfigure}[ht]
    \centering
    \begin{minipage}[b]{0.45\textwidth}
        \centering
        \includegraphics[width=0.7\linewidth]{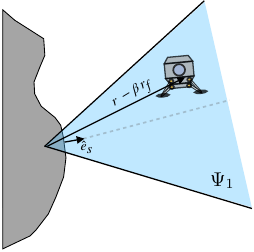}
        \caption{Zero-superlevel set of $\psi_1$.}
        \label{fig:gs_constraint}
    \end{minipage}%
    \hspace{0.01\textwidth}
    \begin{minipage}[b]{0.45\textwidth}
        \centering
        \includegraphics[width=0.7\linewidth]{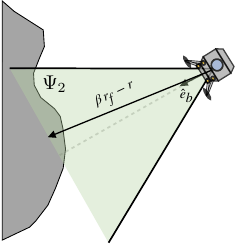}
        \caption{Zero-superlevel set of $\psi_2$.} 
        \label{fig:bs_constraint}
    \end{minipage}
    \hfill
    
    \setcounter{subfigure}{-1}
    \caption{Glideslope and boresight constraints and their zero-superlevel sets.}
    \label{fig:state_constraints}
\end{subfigure}

We also consider constraints on the translational speed and angular speed. 
Let $v_{\rm m}>v_{\rm a}$ be the maximum allowable translational speed, and let $\omega_{\rm m}>\omega_{\rm a}$ be the maximum allowable angular speed. 
Then, consider the barrier functions 
\begin{align}
\psi_3(x) &\triangleq \varphi(\|r-r_{\rm f}\|^2)v_{\rm m}^2 + \left (1 - \varphi(\|r-r_{\rm f}\|^2) \right )v_{\rm a}^2 - \|v\|^2, \label{eq:psi_3}\\
\psi_4(x) &\triangleq \varphi(\|r-r_{\rm f}\|^2)\omega_{\rm m}^2 + \left ( 1 - \varphi(\|r-r_{\rm f}\|^2) \right )\omega_{\rm a}^2 - \|\omega\|^2, \label{eq:psi_4}
\end{align}
where $\varphi: [0, \infty) \to [0,1]$ is 2-times continuously differentiable and nondecreasing such that $\varphi(0)=0$ and $\lim_{a \to \infty} \varphi(a) = 1$. 
As the spacecraft approaches the landing position, $\varphi$ smoothly transitions the speed limits from the maximum allowable values to the allowable values at landing.
The following example provides one possible choice for $\varphi$.

\begin{figure}[t]
    \begin{center}
    \includegraphics[scale=1.1]{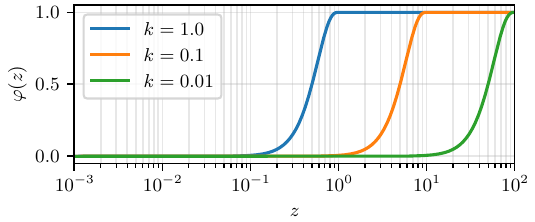}
    \end{center}
    \caption{$\varphi$ given by \Cref{example:transition_func} for 3 values of $k$, which determines the distance at which the speed limits are given by maximum allowable values $v_{\rm m}$ and $\omega_{\rm m}$.}
    \label{fig:Transition_function_example}
\end{figure}

\begin{example}\rm \label{example:transition_func}
Let $k > 0$, and consider
\begin{equation*}
    \varphi(z) = 
    \begin{cases} 
        kz - \frac{\sin 2\pi kz }{2\pi}, & 0 \leq z < \frac{1}{k}, \\
        1, & z \geq \frac{1}{k}.
    \end{cases}
\end{equation*}
\Cref{fig:Transition_function_example} shows this example of $\varphi$ for 3 different values of $k$. 
\end{example}

The set of states that satisfy the translational speed and angular speed constraints are given by 
\begin{equation*}
\Psi_3 \triangleq \{ x \in \BBR^n \colon \psi_3(x) \ge 0 \}, \quad 
\Psi_4 \triangleq \{ x \in \BBR^n \colon \psi_4(x) \ge 0 \}.
\end{equation*}
Thus, the set of states that satisfy all state constraints is given by
\begin{equation*}
    \Psi \triangleq \bigcap_{i\in \{1,2,3,4\} } \Psi_i.
\end{equation*}
Then, the control objectives are summarized as:
\begin{enumO}
    \item \, Solution $x$ lands. \label{Obj:landing} 
    \item \, For all $t \in [0,\infty)$, $x(t) \in \Psi$. \label{Obj:safety}
    \item \, For all $t \in [0,\infty)$, $u(t) \in \SU$. \label{Obj:actuator_input}
\end{enumO}
Together, these objectives imply the spacecraft lands while satisfying all state and actuator input constraints throughout the trajectory despite the uncertainty in the gravitational force. 
The approach in this article can be directly extended to include additional state constraints by introducing additional barrier functions (e.g., $\psi_5,\psi_6,\ldots$) whose zero-superlevel sets model the allowable states. 
We assume that for all $x_0 \in \Psi$ and continuous $\mu$, \cref{eq:system_dynamics,eq:system_dynamics.2,eq:system_dynamics.3,eq:system_dynamics.4} is forward complete.

To accomplish \ref{Obj:landing}--\ref{Obj:actuator_input}, we consider a reference trajectory for the nominal system \cref{eq:system_dynamics,eq:system_dynamics.2,eq:system_dynamics.3,eq:system_dynamics.4,eq:allocation} with $w=0$. 
Specifically, let $x_{\rm r} : [0, \infty) \to \BBR^{n}$ and $u_{\rm r} : [0, \infty) \to \BBR^{m}$ denote the state and control of the reference trajectory.
We assume $(x,u) = (x_{\rm r}, u_{\rm r})$ satisfies the nominal differential equation \cref{eq:system_dynamics,eq:system_dynamics.2,eq:system_dynamics.3,eq:system_dynamics.4,eq:allocation} with $w=0$ and satisfies \ref{Obj:landing}--\ref{Obj:actuator_input}.
In other words, the reference trajectory satisfies the objectives with the nominal gravity model.
However, the reference trajectory generally does not satisfy the objective with nonzero gravity uncertainty.
This nominal open-loop reference trajectory can be generated using a variety of trajectory optimization methods (e.g., \cite{elango_continuous-time_2025,sagliano_six-degree--freedom_2024,reynolds_dual_2020}); however, these trajectory optimization methods are not computationally efficient for recomputing the reference in real time. 
Thus, this article presents a real-time feedback control for satisfying \ref{Obj:landing}--\ref{Obj:actuator_input} by tracking the reference trajectory while deviating as needed to address uncertainty and satisfy constraints. 
We make the technical assumptions that $\mu_\rmr \triangleq M u_\rmr$ is continuously differentiable.

\section{Projection of Actuator Inputs into Control Forces and Torques}
\label{sec:Constraint_Projection}

This section transforms the control variable from actuator inputs $u$ to control forces and torques $\mu$. 
First, we project the set $\SU$ of admissible actuator inputs to obtain the set of admissible control forces and torques, and we develop a computable representation of this set.
Next, we present a quadratic program to obtain the optimal actuator input from control forces and torques.

\subsection{Admissible Control Forces and Torques}

The set of admissible control forces and torques is 
\begin{equation}\label{eq:Phi_def}
    \Phi \triangleq \{ M u : u \in \SU \},
\end{equation}
and note that $u(t) \in \SU$ if and only if $\mu(t) \in \Phi$. 
Hence, \ref{Obj:actuator_input} is equivalent to the condition that for all $t \in [0,\infty)$, $\mu(t) \in \Phi$.

Since $\SU$ is a convex polytope, it follows from \cite[Thm. 19.3]{rockafellar_convex_2015} that $\Phi$ is a convex polytope.
Thus, there exist $a_1, \ldots, a_\ell \in \BBR^6$ and $b_1, \ldots, b_\ell \in \BBR$, where $\ell$ is a positive integer, such that 
\begin{equation}\label{eq:virtual_input_set_final}
    \Phi = \{ \mu \in \BBR^6 : \phi_1(\mu) \geq 0, \ldots, \phi_\ell(\mu)\geq 0\},
\end{equation}
where for $i \in \{1,\ldots,\ell\}$,
\begin{equation}\label{cons:virtual_input_constraints}
    \phi_i(\mu) \triangleq a_i^\top \mu + b_i.
\end{equation}
In other words, \cite[Thm. 19.3]{rockafellar_convex_2015} implies that the half-plane constraints on the actuator input can be transformed into half-plane constraints on the control force and torque. 
\Cref{fig:polytope_proj} illustrates the linear projection of $\SU$ through $M$ to obtain $\Phi$.
In fact, this linear projection can be performed using the beneath-and-beyond algorithm and double-description algorithm in \cite[]{fukuda_double_1996}. 
First, beneath-and-beyond is used to obtain the set of vertices $\SV_u$ of $\SU$ from the half-space representation. 
Then, $\SV_\mu \triangleq \{ Mu \colon u \in \SV_u \}$ contains the set of vertices of $\Phi$. 
Finally, double-description is used to obtain the minimal half-space representation \cref{eq:virtual_input_set_final,cons:virtual_input_constraints} from $\SV_\mu$. 
Note that an upper bound on the number $\ell$ of half-spaces is given by \cite[Thm 8.23]{ziegler_lectures_1995}.

\begin{figure}[ht]
    \centering
    \includegraphics[width=0.25\columnwidth]{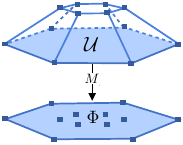}
    \caption{The linear projection of $\SU$ through the affine mapping \cref{eq:allocation} is the set $\Phi$.}
    \label{fig:polytope_proj}
\end{figure}

\subsection{Optimal Actuator Input}

To map from admissible control forces and torques to admissible actuator inputs, consider the point-to-set map $\Upsilon \colon \Phi \rightrightarrows \SU$ given by 
\begin{equation}\label{eq:Upsilon}
    \Upsilon(\mu) \triangleq \{ u \in \SU : \mu = M u \}.
\end{equation}
The following result shows that this set is a convex polytope. 

\begin{proposition}\rm\label{prop:feasibility}
    For all $\mu \in \Phi$, $\Upsilon(\mu)$ is a nonempty convex polytope.
\end{proposition}

\begin{proof}
Let $\mu \in \Phi$, and it follows from \cref{eq:Phi_def} that there exists $u \in \SU$ such that $\mu=Mu$. 
Hence, \cref{eq:Upsilon} implies that $u \in \Upsilon(\mu)$, which implies $\Upsilon(\mu)$ is nonempty. 
Next, since $\Upsilon(\mu)$ is the intersection of the convex polytope $\SU$ with the affine subspace defined by the linear equality constraint \cref{eq:allocation}, it follows that $\Upsilon(\mu)$ is a convex polytope.
\hfill $\blacksquare$
\end{proof}

Next, let $Q \in \BBR^{m \times m}$ be positive definite and $c \in \BBR^m$.
Then, the optimal-actuator-input map is $u_*: \Phi \to \SU$ given by
\begin{equation}\label{eq:opt_allocation}
    u_*(\mu) \triangleq \operatorname*{argmin}_{ \bar{u} \in \Upsilon(\mu)} \bar{u}^{\top}Q\bar{u} + c^\top \bar{u}.
\end{equation}
For all $\mu \in \Phi$, the optimizer $u_*(\mu)$ exists and is unique because $\Upsilon(\mu)$ is nonempty and convex and $Q$ is positive definite. 
Letting $Q=I_m$ and $c=0$ implies that $u_*(\mu)$ is the minimum-norm actuator input in $\SU$ that yields control forces and torques $\mu$.

The quadratic program \cref{eq:opt_allocation} combined with the computable convex polytope \cref{eq:virtual_input_set_final,cons:virtual_input_constraints} allows us to address \ref{Obj:actuator_input} in the space $\Phi$.
Then, \cref{eq:opt_allocation} is solved to obtain the optimal actuator input. 
Hence, the remainder of this article treats $\mu$ as the control variable.

\section{Gravity Estimator, and Desired Control Force and Torque}
\label{sec:FeedbackControl}

This section presents an EHGO to estimate the gravitational uncertainty and uses this estimate to construct the desired control force and torque that achieve \ref{Obj:landing}.

\subsection{Gravity Estimator}

To compensate for gravitational uncertainty, we estimate $w$ using an EHGO method, which is similar to \cite[]{dunham_constrained_2016}. 
Let $\tau > 0$, and let $\nu: [0, \infty) \to \BBR^3$ satisfy 
\begin{equation} \label{eq:EHGO_Dynamics}
        \dot{\nu}(t) = -\tau \Big  [a(x(t)) + \frac{1}{m}R(t)F(t) + \hat{w}(x(t), \nu(t)) \Big],
\end{equation}
where $\nu(0) = \nu_0 \in \BBR^3$ and
\begin{align}
a(x) &\triangleq -2[\omega_{\rm A}]_{\times}v -[\omega_{\rm A}]_{\times}^2r + G_{\rm m}(r), \\
\hat{w}(x, \nu) &\triangleq \tau v + \nu.
\end{align}
Note that $\hat{w}$ is the estimate of $w$, and the error is $\tilde{w}(x, \nu) \triangleq w(x) - \hat{w}(x, \nu)$.
Next, let $\delta: [0,\infty) \to \BBR$ satisfy
\begin{equation}\label{eq:est_error_var_dyn}
\dot{\delta}(t) = -\tau \delta(t) + \frac{c_\delta}{\tau}\|v(t)\|^2,
\end{equation}
where $\delta(0) = \delta_0$ and $c_\delta>0$ is such that for all $r \in \BBR^3$, $c_\delta \geq \|\Delta'(r)\|^2$. 
Next, define
\begin{equation}\label{eq:xi}
    \xi(x, \nu, \delta) \triangleq \delta - \|\tilde{w}(x,\nu)\|^2. 
\end{equation}
The following result shows that $\delta$ is an upper bound on the error $\| \tilde w \|^2$. 
The proof is in \cref{app:Proofs}.

\begin{proposition}\label{prop:Dist_est_conv}\rm
Consider \cref{eq:system_dynamics,eq:system_dynamics.2,eq:system_dynamics.3,eq:system_dynamics.4,eq:EHGO_Dynamics,eq:est_error_var_dyn}, where Assumption~\ref{ass:Bounded_Dist} is satisfied and $\mu$ is continuous. 
Let  $\xi(x_0, \nu_0, \delta_0) \ge 0$.
Then, for all $t\in[0,\infty)$, $\xi(x(t), \nu(t), \delta(t)) \ge 0$. 
\end{proposition}

\Cref{prop:Dist_est_conv} is important because the computable bound $\delta$ is used to guarantee state and input constraints despite the gravity uncertainty.

\subsection{Desired Control Force}

This section presents a desired control force that tracks the reference $(x_{\rm r}, \mu_{\rm r})$ and compensates for the uncertainty $w$. 
Let $x_{\rm r}$ and $\mu_{\rm r}$ be written as 
\begin{equation*}
    x_{\rm r} = \begin{bmatrix} r_{\rm r} \\ v_{\rm r} \\ \mathrm{vec} \,R_{\rm r} \\ \omega_{\rm r} \end{bmatrix}, \qquad \mu_{\rm r} = \begin{bmatrix} F_{\rm r} \\ T_{\rm r} \end{bmatrix},
\end{equation*}
where $r_{\rm r}(t) \in \BBR^3$, $v_{\rm r}(t) \in \BBR^3$, $R_{\rm r}(t) \in \SOT$,  $\omega_{\rm r}(t) \in \BBR^3$, $F_{\rm r}(t) \in \BBR^3$ and $T_{\rm r}(t) \in \BBR^3$.

Define the position and velocity error 
\begin{equation*}
e_r(t,x) \triangleq r- r_{\rm r}(t),\qquad e_v(t,x) = v - v_{\rm r}(t).
\end{equation*}
Then, consider the desired control force
\begin{equation}
F_{\rm d}(t, x, \nu) \triangleq m R^\top \Big( -a(x) - \hat{w}(x,\nu) + a(x_{\rm r}(t)) + \frac{1}{m}R_{\rm r}(t)F_{\rm r}(t)- K_ve_v(t,x) - K_re_r(t,x) \Big),\label{eq:Fd}
\end{equation}
where $K_{r}, K_v \in \BBR^{3 \times 3}$ are positive definite.
The following result provides convergence properties for $e_r$ and $e_v$. 
The proof is in \cref{app:Proofs}.

\begin{proposition}\rm\label{lemma:position_control_stability}
Consider \cref{eq:system_dynamics,eq:system_dynamics.2,eq:system_dynamics.3,eq:system_dynamics.4,eq:EHGO_Dynamics,eq:est_error_var_dyn}, where Assumption~\ref{ass:Bounded_Dist} is satisfied, $T$ is continuous, and $F = F_{\rm d}$. 
Let $\xi(x_0, \nu_0, \delta_0) \geq 0$, and assume that for all $t\geq0$, $\|v(t)\| \leq v_{\rm m}$.
Then, there exists $\vartheta > 0$ such that
    \begin{equation*}
        \left \| \begin{bmatrix} e_r(t,x(t)) \\ e_v(t,x(t)) \end{bmatrix} \right \| \to \left [ 0, \frac{\vartheta}{\tau} \right ] \quad \text{as} \quad t \to \infty.
    \end{equation*}
\end{proposition}

\Cref{lemma:position_control_stability} shows that the desired control force \cref{eq:Fd} yields arbitrarily small steady-state error based on design of the EHGO gain $\tau$.

\subsection{Desired Control Torque}

This section presents a desired control torque that tracks the reference $(x_{\rm r}, \mu_{\rm r})$. 
Define the rotation and angular velocity errors
\begin{equation*}
    e_{R}(t,x) \triangleq \frac{1}{2\sqrt{1 + \mathrm{tr}[R_{\rmr}(t)^\top R]}}(R_{\rmr}(t)^\top R - R^\top R_{\rmr}(t))^\vee, \qquad
    e_{\omega}(t,x) \triangleq \omega - R^\top R_{\rmr}(t)\omega_{\rmr}(t).
\end{equation*}
Then, consider the desired control torque 
\begin{equation}\label{eq:Td}
    T_{\rm d}(t, x) \triangleq J_B \Big( [R^{\top}R_{\rm r}(t)\omega_{\rm r}(t)]_{\times} J_B R^{\top}R_{\rm r}(t)\omega_{\rm r}(t) + J_B R^{\top}R_{\rm r}(t)\dot{\omega}_{\rm r}(t) -K_R e_R(t,x) - K_{\omega}e_{\omega}(t,x) \Big),
\end{equation}
where $K_R,K_\omega \in \BBR^{3 \times 3}$ are positive definite.
The next result is from \cite[Proposition 3]{lee_exponential_2012} and demonstrates that the errors converge to zero.

\begin{proposition}\rm\label{lemma:attitude_control_stability}
Consider \cref{eq:Attitude_dyn,eq:angvel_dyn}, where $T = T_{\rm d}$. 
Then, $\lim_{t \to \infty} e_R(t,x(t)) =0$ and $\lim_{t \to \infty} e_\omega(t,x(t)) =0$. 
\end{proposition}

Let $\mu_{\rm d}: [0,\infty) \times \BBR^{n} \times \BBR^3 \to \BBR^6$ be given by
\begin{equation}
    \mu_{\rm d}(t, x, \nu) \triangleq \begin{bmatrix}
    F_{\rm d}(t, x, \nu)\\
    T_{\rm d}(t, x)
    \end{bmatrix}. \label{eq:ud}
\end{equation}
\Cref{lemma:position_control_stability,lemma:attitude_control_stability} imply that for a sufficiently large $\tau$, \ref{Obj:landing} is satisfied with  $\mu=\mu_\rmd$.
However, $\mu_{\rm d}$ does not necessarily satisfy \ref{Obj:safety}.
Furthermore, $\mu_{\rm d}$ is not necessarily in $\Phi$, which implies that it cannot necessarily be generated by a control $u$ that satisfies \ref{Obj:actuator_input}.

\section{Optimal Control for Safe Landing} \label{sec:CBFControl}

This section presents a minimum-intervention control to satisfy \ref{Obj:safety} and \ref{Obj:actuator_input}. We first derive higher-order barrier functions to express both objectives as forward invariance of a set. Next, we construct a composite CBF which is then used to compute a control that ensures forward invariance of such set.

\subsection{Control Dynamics and Barrier Functions}

Since $\phi_1,\ldots,\phi_{\ell}$ depend explicitly on $\mu$, they are not valid candidate CBFs. 
We adopt the method in \cite[]{rabiee_closed-form_2024,rabiee_composition_2024} to transform these input constraints into controller-state constraints by introducing control dynamics. 
Let the control force and torque $\mu$ be the solution to the integrator
\begin{equation}\label{eq:control_dynamics}
    \dot{\mu}(t) = \zeta(t),
\end{equation}
where $\mu(0) = \mu_0$ is the initial condition and $\zeta: [0,\infty) \to \BBR^6$ is the surrogate control, which is the input to the integrator \cref{eq:control_dynamics}.

The cascade of 
\cref{eq:system_dynamics,eq:system_dynamics.2,eq:system_dynamics.3,eq:system_dynamics.4,eq:EHGO_Dynamics,eq:est_error_var_dyn,eq:EHGO_Dynamics,eq:est_error_var_dyn,eq:control_dynamics} of all system states can be expressed as 
\begin{equation}\label{eq:Augmented_Dynamics}
    \dot{\tilde{x}} = \tilde{f}(\tilde{x})+\tilde{B}\zeta + \tilde{D}w(x)
\end{equation}
where $\tilde{n}= 28$ and
\begin{gather}
\tilde{x} \triangleq \begin{bmatrix} x \\ \mu \\ \nu \\ \delta \end{bmatrix},\quad
\tilde{x}_0 \triangleq \begin{bmatrix} x_0 \\ \mu_0 \\ \nu_0 \\ \delta_0 \end{bmatrix}, \label{eq:Augmented_Dynamics_1}\\
\tilde{f}(\tilde{x}) \triangleq \begin{bmatrix}
        f(x) + g(x)\mu \\
        0_{6 \times 1} \\
        -\tau \Big  [a(x) + \frac{1}{m}RF + \hat{w}(x, \nu) \Big] \\
        -\tau \delta + \frac{c_\delta}{\tau}\|v\|^2
    \end{bmatrix},\quad
\tilde{B} \triangleq \begin{bmatrix} 0_{18 \times 6} \\ I_{6 \times 6} \\ 0_{3 \times 6} \\ 0_{1 \times 6} \end{bmatrix}, \quad
\tilde{D} \triangleq\begin{bmatrix}D \\ 0_{6 \times 3} \\ 0_{3 \times 3} \\ 0_{1 \times 3}\end{bmatrix}.
\label{eq:Augmented_Dynamics_2}
\end{gather}

Since $\mu$ is the state of \cref{eq:control_dynamics}, the input constraints are now state constraints of the cascade \cref{eq:Augmented_Dynamics,eq:Augmented_Dynamics_1,eq:Augmented_Dynamics_2}. 
Note that $\phi_1,\ldots,\phi_{\ell}$ have relative degree one with respect to \cref{eq:Augmented_Dynamics,eq:Augmented_Dynamics_1,eq:Augmented_Dynamics_2} with input $\zeta$. 
In contrast, $\psi_1$ and $\psi_2$ have relative degree 3 with respect to \cref{eq:Augmented_Dynamics,eq:Augmented_Dynamics_1,eq:Augmented_Dynamics_2} with input $\zeta$, and $\psi_3$ and $\psi_4$ have relative degree 2.
Thus, we adopt a higher-order approach (e.g., \cite{tan_high-order_2022}) to construct higher-order candidate CBFs.

First, we address $\psi_1$ and $\psi_2$, which are relative degree 3 with respect to $\zeta$.
Let $\beta_{1},\beta_{2},\gamma_{1},\gamma_{2} \colon \BBR \to \BBR$ be extended class-$\mathcal{K}$ functions, where $\beta_{1},\beta_{2}$ are 3-times continuously differentiable and $\gamma_1,\gamma_2$ are 2-times continuously differentiable. 
Then, for $i \in \{1,2\}$, consider $\psi_{i,1},\psi_{i,2} \colon \BBR^{\tilde n} \to \BBR$ defined by
\begin{align}
\psi_{i,1}(\tilde{x}) &\triangleq L_{\tilde{f}} \psi_{i}(\tilde{x}) + \beta_{i}(\psi_{i}(\tilde{x})),\\
\psi_{i,2}(\tilde{x}) &\triangleq L_{\tilde{f}} \psi_{i,1}(\tilde{x}) + \psi_{i,1}^\prime(\tilde{x}) \tilde D w(x) + \gamma_{i}(\psi_{i,1}(\tilde{x})).\label{eq:psi_i2}
\end{align}
Next, we address $\psi_3$ and $\psi_4$, which are relative degree~2 with respect to $\zeta$.
Let $\beta_{3},\beta_{4} \colon \BBR \to \BBR$ be extended class-$\mathcal{K}$ and 2-times continuously differentiable functions. 
Then, for $i \in \{3,4\}$, consider $\psi_{i,1} \colon \BBR^{\tilde n} \to \BBR$ defined by
\begin{equation}
\psi_{i,1}(\tilde{x}) \triangleq L_{\tilde{f}} \psi_{i}(\tilde{x}) + \psi_{i}^\prime(\tilde{x}) \tilde D w(x) + \beta_{i}(\psi_{i}(\tilde{x})).\label{eq:psi_i1}
\end{equation}

The intersection of the zero-superlevel sets of the higher-order barrier functions is given by
\begin{align}
\SP &\triangleq \{ \tilde x \in \BBR^{\tilde n} \colon \, \psi_{1,2}(\tilde{x}) \ge 0, \psi_{2,2}(\tilde{x}) \ge 0, \psi_{3,1}(\tilde{x}) \ge 0,  \psi_{4,1}(\tilde{x}) \ge 0, \phi_1(\mu) \ge 0, \ldots, \phi_{\ell}(\mu) \ge 0 \},
\end{align}
and the intersections of all barrier functions is 
\begin{align}
\SC &\triangleq \{ \tilde x \in \SP \colon \, \psi_{1}(x) \ge 0, \psi_{2}(x) \ge 0, \psi_{3}(x) \ge 0, \psi_{4}(x) \ge 0,  \psi_{1,1}(\tilde x) \ge 0, \psi_{2,1}(\tilde x) \ge 0 \}. \label{eq:C_set}
\end{align}
It follows from \cite[Prop. 1]{tan_high-order_2022} that if $\tilde{x}_0 \in \SC$ and for all $t \in [0, \infty)$, $\tilde{x}(t) \in \SP$, then for all $t \in [0, \infty)$, $\tilde{x}(t) \in \SC$, which implies that $x(t) \in \Psi$ and $\mu(t) \in \Phi$. 
Thus, satisfying \ref{Obj:safety} and \ref{Obj:actuator_input}.
However, $\SP$ is constructed from the zero-superlevel set of \cref{eq:psi_i2,eq:psi_i1}, which are not implementable because they depend on the unknown function $w$. 
Thus, we use estimate $\hat w$ and bound $\delta$ to construct implementable lower bounds. 
Specifically, for $i \in \{1,2\}$, consider $\ubar{\psi}_{i,2} \colon \BBR^{\tilde n} \to \BBR$ given by
\begin{align}
\ubar{\psi}_{i,2}(\tilde{x}) &\triangleq L_{\tilde{f}} \psi_{i,1}(\tilde{x}) + \psi_{i,1}^\prime(\tilde{x}) \tilde D \hat w(x,\nu) + \gamma_{i}(\psi_{i,1}(\tilde{x})) - \sqrt{\varepsilon + \delta \| \tilde D^\top \psi_{i,1}^\prime(\tilde{x})^\top \|^2}, \label{eq:bar psi_i2}
\end{align}
and for $i \in \{3,4\}$, consider $\ubar{\psi}_{i,1} \colon \BBR^{\tilde n} \to \BBR$ given by
\begin{align}
\ubar{\psi}_{i,1}(\tilde{x}) &\triangleq L_{\tilde{f}} \psi_{i}(\tilde{x}) + \psi_{i}^\prime(\tilde{x}) \tilde D \hat w(x,\nu) + \beta_{i}(\psi_{i}(\tilde{x})) - \sqrt{\varepsilon + \delta \| \tilde D^\top \psi_{i}^\prime(\tilde{x})^\top \|^2},\label{eq:bar psi_i1}
\end{align}
where $\varepsilon > 0$. 
The next result shows that the computable functions \cref{eq:bar psi_i2,eq:bar psi_i1} are lower bounds for \cref{eq:psi_i2,eq:psi_i1}.

\begin{proposition}\rm 
\label{prop:computable_lower_bounds}
Let $\tilde{x} \in \BBR^{\tilde n}$ be such that $\xi(x,\nu,\delta) \geq 0$. 
Then, for $i \in \{1,2\}$, $\ubar{\psi}_{i,2}(\tilde{x}) \leq \psi_{i,2}(\tilde{x})$, and for $i \in \{3,4\}$, $\ubar{\psi}_{i,1}(\tilde{x}) \leq \psi_{i,1}(\tilde{x})$.
\end{proposition}

\begin{proof}
Let $i \in \{1,2\}$, and it follows from \cref{eq:psi_i2,eq:bar psi_i2} that
\begin{align*}
    \psi_{i,2}(\tilde{x}) -\ubar{\psi}_{i,2}(\tilde{x}) & = \psi_{i,1}' \tilde{D} \big(w(x) - \hat w(x,\nu)
    \big) + \sqrt{\varepsilon + \delta \| \tilde D^\top \psi_{i,1}^\prime(\tilde{x})^\top \|^2} \\
    & \geq \psi_{i,1}' \tilde{D} \tilde{w}(x,\nu) + \sqrt{\delta} \| \tilde D^\top \psi_{i,1}^\prime(\tilde{x})^\top \| \\
    & \geq - \| \tilde D^\top \psi_{i,1}^\prime(\tilde{x})^\top \| \|\tilde{w}(x,\nu)\| + \sqrt{\delta} \| \tilde D^\top \psi_{i,1}^\prime(\tilde{x})^\top \| \\
    & = \| \tilde D^\top \psi_{i,1}^\prime(\tilde{x})^\top \| \big(\sqrt{\delta} -\|\tilde{w}(x,\nu)\| \big),
\end{align*}
which is nonnegative because $\xi(x, \nu. \delta) \geq 0$.
Using the same process for $i \in \{3,4\}$ we show that $\psi_{i,1}(\tilde{x}) \geq \ubar{\psi}_{i,1}(\tilde{x})$, which completes the proof.
\hfill $\blacksquare$
\end{proof}

Since \Cref{prop:computable_lower_bounds} implies $\xi(x,\nu,\delta) \geq 0$ is sufficient for \cref{eq:bar psi_i2,eq:bar psi_i1} to lower bound \cref{eq:psi_i2,eq:psi_i1}, define 
\begin{equation}
    \Xi \triangleq \{ \tilde x \in \BBR^{\tilde n} \colon \, \xi(x,\nu,\delta) \geq 0 \},
\end{equation}
and note that \Cref{prop:Dist_est_conv} implies that $\Xi$ is forward invariant. 
Next, define
\begin{align}
\ubar{\SP} &\triangleq \{ \tilde x \in \BBR^{\tilde n} \colon \, \ubar{\psi}_{1,2}(\tilde{x}) \ge 0, \ubar{\psi}_{2,2}(\tilde{x}) \ge 0, \ubar{\psi}_{3,1}(\tilde{x}) \ge 0,  \ubar{\psi}_{4,1}(\tilde{x}) \ge 0, \phi_1(\mu) \ge 0, \ldots, \phi_{\ell}(\mu) \ge 0 \},\label{eq:p_bar_set}
\end{align}
and 
\begin{equation}
    \ubar{\SC} \triangleq \SC \cap \Xi. \label{eq:C_bar_set}
\end{equation}
The next result shows that forward invariance of $\ubar{\SP}$ can ensure forward invariance of $\ubar{\SC}$.

\begin{proposition}\rm 
\label{prop.forward_invariance}
Assume Assumption~\ref{ass:Bounded_Dist} is satisfied.
Let $\tilde x_0 \in \ubar{\SC}$ and $t_{\rm m} \in (0,\infty]$. Assume for all $t \in [0, t_{\rm m})$, $\tilde{x}(t) \in \ubar{\SP}$.
Then, for all $t \in [0, t_{\rm m})$, $\tilde{x}(t) \in \ubar{\SC}$.
\end{proposition}

\begin{proof}
Let $t_1 \in [0, t_{\rm m})$. 
Since $\tilde x_0 \in \ubar{\SC} = \SC \cap \Xi$, \Cref{prop:Dist_est_conv} implies that $\tilde x(t_1) \in \Xi$.
Since, in addition, $\tilde x(t_1) \in \ubar{\SP}$, it follows from \Cref{prop:computable_lower_bounds} that $\tilde x(t_1) \in \SP$.
Thus, for all $t \in [0,t_\rmm)$, $\tilde x(t) \in \Xi \cap \SP$.
Since $\tilde{x}_0 \in \SC$ and for all $t \in [0,t_\rmm)$, $\tilde{x}(t) \in \SP$, it follows from \cite[Prop. 1]{tan_high-order_2022} that for all $t \in [0,t_\rmm)$,  $\tilde{x}(t) \in \SC$, which confirms that $\tilde x(t) \in \SC \cap \Xi$.
\hfill $\blacksquare$
\end{proof}

\Cref{prop.forward_invariance} shows that forward invariance of $\ubar{\SP}$ is sufficient for \ref{Obj:safety} and \ref{Obj:actuator_input}. 
Motivated by this fact, the next section constructs a candidate CBF whose zero-superlevel set is a subset of $\ubar{\SP}$.

\subsection{Soft-Minimum Relaxed CBF}

Let $\rho,k_1, \ldots, k_{\ell+4} > 0$, and consider $h: \BBR^{\tilde n} \to \BBR$ defined by
\begin{equation} \label{eq:h}
    h(\tilde{x}) \triangleq \mathrm{softmin}_\rho \Big( k_1\ubar{\psi}_{1,2}(\tilde{x}), k_2\ubar{\psi}_{2,2}(\tilde{x}), k_3\ubar{\psi}_{3,1}(\tilde{x}), 
    k_4 \ubar{\psi}_{4,1}(\tilde{x}), k_5\phi_1(\mu), \ldots, k_{\ell+4}\phi_{\ell}(\mu)  \Big ),
\end{equation}
where 
\begin{equation}
    \mathrm{softmin}_\rho ( z_1, \ldots, z_N) \triangleq - \frac{1}{\rho} \mathrm{log} \sum_{i=1}^N e^{-\rho z_i} \label{eq:softmin}
\end{equation}
is the log-sum-exponential \textit{soft minimum}. 
The zero-superlevel set of $h$ is
\begin{equation}\label{eq:h_set}
    \mathcal{H} \triangleq \{ \tilde{x} \in \BBR^{\tilde n}: h(\tilde{x}) \geq 0 \}.
\end{equation}
It follows from \cite[Prop. 2]{safari_time-varying_2025} that $\mathcal{H} \subset \ubar{\mathcal{P}}$.
In fact, as $\rho \to \infty, \mathcal{H} \to \ubar{\mathcal{P}}$.
The soft minimum \eqref{eq:softmin} is a continuously differentiable lower bound on the minimum.
The worst case conservativeness of the soft minimum in comparison to the minimum is $(\log N)/\rho$ \cite[]{rabiee_soft_barriers_2025}. 
Hence, $\rho$ can be selected sufficiently large to limit conservativeness based on $N$ and the order of magnitude of the arguments of the soft minimum. 
However, selecting $\rho$ is a trade off between conservativeness of the soft-minimum approximation and the magnitude of $h^\prime$ at points where the minimum is not differentiable \cite[]{rabiee_soft_barriers_2025}. 
The positive constants $k_1, \ldots, k_{\ell+4} > 0$ are selected to normalize the magnitude of the arguments of \cref{eq:h}.

Let $\alpha \colon \BBR \to \BBR$ be locally Lipschitz and nondecreasing such that $\alpha(0)=0$. 
Then, consider the constraint function $b \colon \BBR^{\tilde n} \times \BBR^6 \times \BBR \to \BBR$ given by
\begin{equation}\label{eq:b}
    b(\tilde{x}, \bar{\zeta}, \bar{\kappa}) \triangleq h^\prime(\tilde x) \left [ \tilde{f}(\tilde{x})+\tilde{B} \bar \zeta + \tilde{D}w(x) \right ] + \alpha\big(h(\tilde{x})\big) 
    + \bar{\kappa} h(\tilde{x}),
\end{equation}
where $\bar \zeta$ is the surrogate control variable and $\bar \kappa$ is the slack variable. 
Note that $b(\tilde{x}, \bar{\zeta}, \bar{\kappa}) \geq 0$ is a CBF constraint that can be used to select a surrogate control that makes $\SH$ forward invariant.
However, $b$ depends on the unknown function $w$, which implies that $b(\tilde x,\bar{\zeta},\bar{\kappa}) \geq 0$ cannot be implemented. 
Thus, we use estimate $\hat w$ and bound $\delta$ to construct the implementable lower bound $\ubar{b} \colon \BBR^{\tilde n} \times \BBR^6 \times \BBR \to \BBR$ defined as
\begin{equation}\label{eq:b_ubar}
    \ubar{b}(\tilde{x}, \bar{\zeta}, \bar{\kappa}) \triangleq h^\prime(\tilde x) \left [ \tilde{f}(\tilde{x})+\tilde{B} \bar \zeta + \tilde{D} \hat w(x,\nu) \right ]
    - \sqrt{\delta}\|\tilde{D}^\top h^\prime(\tilde{x})^\top\| 
    + \alpha\big(h(\tilde{x})\big) 
    + \bar{\kappa} h(\tilde{x}).
\end{equation}
The next result shows that $\ubar{b}$ is a lower bound on $b$. 
The proof is similar to that of \cref{prop:computable_lower_bounds}.

\begin{proposition}\rm\label{prop:ubar_b_lower_b}
    Let $(\tilde x, \bar{\zeta} ,\bar{\kappa}) \in \Xi \times \BBR^6 \times \BBR$.
    Then, $\ubar{b}(\tilde{x}, \bar{\zeta}, \bar{\kappa}) \leq b(\tilde{x}, \bar{\zeta}, \bar{\kappa})$.
\end{proposition}

Since $b(\tilde{x}, \bar{\zeta}, \bar{\kappa}) \geq 0$ can be used to select a surrogate control that makes $\SH \subset \ubar \SP$ forward invariant, it follows from \Cref{prop.forward_invariance,prop:ubar_b_lower_b} that $\ubar b(\tilde x,\bar{\zeta},\bar{\kappa}) \geq 0$ can be used to select a surrogate control that makes $\ubar \SC$ forward invariant, thus satisfying \ref{Obj:safety} and \ref{Obj:actuator_input}. 
The next result provides a sufficient condition such that $\ubar b(\tilde x,\bar{\zeta},\bar{\kappa}) \geq 0$ is feasible for all $\tilde x \in \BBR^{\tilde n}$.

\begin{proposition}\rm \label{prop:ubar_b_feasibility}
Assume that for all $\tilde x \in \mbox{bd } \SH$, $\frac{\partial h}{\partial \mu}(\tilde x) \neq 0$. 
Then, for all $\tilde x \in \BBR^{\tilde n}$, there exist $(\bar \zeta ,\bar \kappa ) \in \BBR^6 \times \BBR$ such that $\ubar{b}(\tilde x, \bar \zeta, \bar \kappa) \geq 0$.
\end{proposition}

\begin{proof}
We consider 2 cases: $\tilde x \in \mbox{bd } \mathcal{H}$, and $\tilde{x}  \in \BBR^{\tilde n} \setminus \mbox{bd } \mathcal{H}$.
First, let $\tilde x \in \mbox{bd } \mathcal{H}$. 
Let $\bar \kappa \in \BBR$ and
    \begin{equation*}
        \bar \zeta = \frac{- h^\prime(\tilde x) \left [ \tilde{f}(\tilde{x}) + \tilde{D} \hat w(x,\nu) \right ] + \sqrt{\delta}\Vert{}\tilde{D}^\top h^\prime(\tilde{x})^\top\Vert{} }{\left \|\frac{\partial h}{\partial \mu}(\tilde x) \right \|^2}
        \left [ \frac{\partial h}{\partial \mu}(\tilde x) \right ]^\top
    \end{equation*}
Then, \cref{eq:b_ubar} yields $\ubar{b}(\tilde x, \bar \zeta, \bar \kappa) = 0$.    
    
Next, let $\tilde{x}  \in \BBR^{\tilde n} \setminus \mbox{bd } \mathcal{H}$. 
Let $\bar \zeta \in \BBR^6$ and
    \begin{equation*}
        \bar \kappa = \frac{-  h^\prime(\tilde x) \left [ \tilde{f}(\tilde{x})+\tilde{B} \bar \zeta + \tilde{D} \hat w(x,\nu) \right ]
    + \sqrt{\delta}\|\tilde{D}^\top h^\prime(\tilde{x})^\top\| 
    - \alpha\big(h(\tilde{x})\big)}{h(\tilde{x})}.
    \end{equation*}
Then, \cref{eq:b_ubar} yields $\ubar{b}(\tilde x, \bar \zeta, \bar \kappa) = 0$.
    \hfill $\blacksquare$
\end{proof}

Based on \Cref{prop:ubar_b_feasibility}, we make the following assumption. 

\begin{assumption}\label{ass:Lgh}
\rm
For all $\tilde x \in \mbox{bd } \SH$, $\frac{\partial h}{\partial \mu}(\tilde x) \neq 0$.
\end{assumption}

\subsection{Closed-Form Optimal Control}

This section presents a real-time feedback control that satisfies \ref{Obj:landing}--\ref{Obj:actuator_input}. 
First, consider the desired surrogate control $\zeta_{\rm d} \colon [0,\infty) \times \BBR^{\tilde n} \to \BBR^{6}$ defined by
\begin{equation} \label{eq:zeta_d}
    \zeta_{\rm d}(t, \tilde x) \triangleq \dot{\mu}_{\rm d}(t, x, \nu) - \sigma \big(\mu - \mu_{\rm d}(t, x, \nu) \big),
\end{equation}
where $\sigma > 0$. 
If $\zeta = \zeta_{\rm d}$, then the solution $\mu$ to the integrator \eqref{eq:control_dynamics} converges exponentially to $\mu_{\rm d}$ \cite[Prop. 8]{rabiee_closed-form_2024}. 
Since the integral of the desired surrogate control $\zeta_{\rm d}$ converges exponentially to the desired control force and torque $\mu_\rmd$, it follows from \Cref{lemma:position_control_stability,lemma:attitude_control_stability} that for a sufficiently large $\tau$, \ref{Obj:landing} is satisfied for $\zeta = \zeta_{\rm d}$.
However, $\mu_{\rm d}$ does not necessarily satisfy \ref{Obj:safety}, and $\mu_{\rm d}$ cannot necessarily be generated by a control $u$ that satisfies \ref{Obj:actuator_input}.
Thus, we consider the minimum-intervention cost function 
\begin{equation} \label{eq:Cost_function}
    J(t, \tilde{x}, \bar{\zeta}, \bar{\kappa}) \triangleq \| \bar{\zeta} - \zeta_{\rm d}(t, \tilde{x}) \|^2 + \gamma \bar{\kappa}^2,
\end{equation}
where $\gamma>0$. 
The objective is to find $(\bar{\zeta}, \bar{\kappa})$ That minimize $J$ subject to the constraint $\ubar{b}(\tilde{x}, \bar{\zeta}, \bar{\kappa}) \ge 0$.

For each $\tilde{x} \in \BBR^{\tilde n}$, the minimizer of \cref{eq:Cost_function} subject to $\ubar{b}(\tilde{x}, \bar{\zeta}, \bar{\kappa}) \geq 0$ can be derived from the first-order necessary conditions for optimality (e.g., \cite{rabiee_closed-form_2024,safari_time-varying_2025,Wieland_CBF_2007}).
Specifically, the optimal surrogate control $\zeta_* \colon [0,\infty) \times \BBR^{\tilde n} \to \BBR^6$ is given by 
\begin{equation}\label{eq:safe_control}
    \zeta_*(t, \tilde{x}) \triangleq \zeta_{\rm d}(t, \tilde{x}) +\lambda(t, \tilde{x})\frac{\partial h}{\partial \mu}(\tilde x)^\top,
\end{equation}
where $\lambda \colon [0,\infty) \times \BBR^{\tilde n} \to \BBR$ is defined by
\begin{gather}
\lambda(t, \tilde{x}) \triangleq \max \left\{0, \frac{- \ubar{b}(\tilde{x}, \zeta_{\rm d}(t,\tilde{x}),0)}{\|\frac{\partial h}{\partial \mu}(\tilde x)^\top\|^2 +\gamma^{-1}h(\tilde{x})^2} \right \}.\label{eq:lambda_safety}
\end{gather}
In addition, the optimal slack variable $\kappa_* \colon [0,\infty) \times \BBR^{\tilde n} \to \BBR$ is given by 
\begin{equation}\label{eq:safe_kappa}
    \kappa_*(t, \tilde{x}) \triangleq \gamma^{-1}h(\tilde{x})\lambda(t, \tilde{x}).
\end{equation}
The parameter $\gamma$ weights $\kappa$ in the cost \eqref{eq:Cost_function} and influences control aggressiveness. 
The next result shows that the control $\zeta_*$ and $\kappa_*$ satisfy the safety constraint $\ubar{b}(x, \zeta_*,\kappa_*) \geq 0$, and provides useful properties regarding the continuity of $\zeta^*$, $\kappa^*$ and $\lambda$.

\begin{proposition}\rm\label{prop:safe_control_properties}
    If Assumption~\ref{ass:Lgh} is satisfied, then, the following statements hold:
    \begin{enumerate}
    \renewcommand{\theenumi}{\alph{enumi}}
    \renewcommand{\labelenumi}{\indent (\theenumi)}
    
        \item For all $(t,x) \in [0,\infty) \times \BBR^{\tilde n}$, $\ubar{b}(x, \zeta_*(t,\tilde{x}),\kappa_*(t, \tilde{x})) \geq 0$. \label{prop:safe_b}
        
        \item $\zeta_*$, $\kappa_*$ and $\lambda$ are continuous on $[0,\infty) \times \BBR^{\tilde{n}}$. \label{prop:safe_c}
        
        \item Assume that $h'$ is locally Lipschitz on $\BBR^{\tilde{n}}$. Then, $\zeta_*$, $\kappa_*$ and $\lambda$ are locally Lipschitz on $[0,\infty) \times \BBR^{\tilde{n}}$. \label{prop:safe_d}
    \end{enumerate}
\end{proposition}

\begin{proof}
    To prove (\ref{prop:safe_b}), substituting \cref{eq:safe_control,eq:safe_kappa} into \cref{eq:b_ubar} yields
    \begin{align*}
        \ubar{b}(\tilde{x}, \zeta_{*}(t,\tilde{x}),\kappa_*(t,\tilde{x})) &= h^\prime(\tilde x) \left [ \tilde{f}(\tilde{x})+\tilde{B} \Big (\zeta_{\rm d}(t, \tilde{x}) +\lambda(t, \tilde{x})\frac{\partial h}{\partial \mu}(\tilde x)^\top \Big ) + \tilde{D} \hat w(x,\nu) \right ] \\
        &\qquad
    - \sqrt{\delta}\|\tilde{D}^\top h^\prime(\tilde{x})^\top\| 
    + \alpha\big(h(\tilde{x})\big) 
    + \gamma^{-1}\lambda(t, \tilde{x}) h(\tilde{x})^2\\
    &= \ubar{b}(\tilde{x}, \zeta_{\rm d}(t,\tilde{x}),0) + \lambda(t,x) \Big( \|\frac{\partial h}{\partial \mu}(\tilde x)^\top\|^2 +\gamma^{-1}h(\tilde{x})^2 \Big) \\
    & = \max\{0 , \ubar{b}(\tilde{x}, \zeta_{\rm d}(t,\tilde{x}),0)\} \\
    & \geq 0.
    \end{align*}

    To prove (\ref{prop:safe_c}), since \cref{cons:virtual_input_constraints,eq:bar psi_i1,eq:bar psi_i2} are continuously differentiable on $\BBR^{\tilde n}$, it follows from \cref{eq:softmin,eq:h} that $h$ is continuously differentiable on $\BBR^{\tilde n}$.
    Thus, $h'$ is continuous on $\BBR^{\tilde n}$.
    Furthermore, since $\mu_{\rm r}$ is continuously differentiable on $[0,\infty)$, it follows from \cref{eq:ud,eq:control_dynamics,eq:zeta_d} that $\zeta_{\rm d}$ is continuous on $[0,\infty) \times \BBR^{\tilde n}$.
    Since, in addition, the right hand side of \cref{eq:EHGO_Dynamics,eq:est_error_var_dyn} are locally Lipschitz on $\BBR^{\tilde n}$, it follows that $\hat{w}$ and $\delta$ are continuous on $[0,\infty)$.
    Hence, it follows from \cref{eq:b_ubar} that $\ubar{b}(\tilde{x}, \zeta_{\rm d}(t,\tilde{x}),0)$ is continuous on $[0,\infty) \times \BBR^{\tilde n}$.
    Moreover, it follows from \cref{eq:lambda_safety} and Assumption~\ref{ass:Lgh} that $\lambda$ is continuous on $[0,\infty) \times \BBR^{\tilde n}$, which combined with \cref{eq:safe_control,eq:safe_kappa} implies that $\zeta_*$ and $\kappa_*$ are continuous on $[0,\infty) \times \BBR^{\tilde n}$.

    To prove (\ref{prop:safe_d}), note that $\tilde{f}$ is locally Lipschitz on $\BBR^{\tilde n}$ and $\zeta_{\rm d}$ is locally Lipschitz in $\tilde x$ on $\BBR^{\tilde n}$.
    Since, in addition, $h$ and $h'$ are locally Lipschitz on $\BBR^{\tilde n}$, it follows from \cref{eq:b_ubar,eq:lambda_safety} that $\lambda$ is locally Lipschitz on $\BBR^{\tilde n}$, which combined with \cref{eq:safe_control,eq:safe_kappa} implies that $\zeta_*$ and $\kappa_*$ are locally Lipschitz on $\BBR^{\tilde n}$.
    \hfill $\blacksquare$
\end{proof}

The following theorem is the main result on constraint satisfaction.

\begin{theorem}\rm{ \label{thm:safety}
Consider \cref{eq:Augmented_Dynamics,eq:Augmented_Dynamics_1,eq:Augmented_Dynamics_2}, where Assumptions~\ref{ass:Bounded_Dist} and~\ref{ass:Lgh} are satisfied. 
Let $\zeta = \zeta_*$, where $\zeta_*$ is given by \cref{eq:safe_control,eq:lambda_safety}. 
Assume that $h^\prime$ is locally Lipschitz on $\BBR^{\tilde n}$. 
Then, for all $\tilde{x}_0 \in \ubar{\mathcal{C}}$, the following statements hold:
    \begin{enumerate}
    \renewcommand{\theenumi}{\alph{enumi}}
    \renewcommand{\labelenumi}{\indent (\theenumi)}

    \item 
    There exists a maximum $t_{\rm m}\in(0,\infty]$ such that \cref{eq:Augmented_Dynamics,eq:Augmented_Dynamics_1,eq:Augmented_Dynamics_2} with $\zeta=\zeta_*$ has a unique solution on $[0,t_{\rm m})$. \label{thm:safety_a}

    \item 
    For all $t \in [0,t_{\rm m})$, $\tilde{x}(t) \in \ubar{\mathcal{C}}$. \label{thm:safety_b}

    \item 
    For all $t \in [0,t_{\rm m})$, $x(t) \in \Psi$ and $\mu(t) \in \Phi$. \label{thm:safety_c}
    
    \item 
    If $t_{\rm m} = \infty$, then $v$, $R$, $\omega$, $\nu$ and $\delta$ are bounded. \label{thm:safety_d}

    \end{enumerate}
}
\end{theorem}

\begin{proof}
To prove (\ref{thm:safety_a}), it follows from (\ref{prop:safe_c}) and (\ref{prop:safe_d}) of \cref{prop:safe_control_properties} that $\zeta_*$ is continuous in $t$ and locally Lipschitz in $\tilde x$. 
Since, in addition, $\tilde{f}$ is locally Lipschitz on $\BBR^{\tilde n}$, it follows from \cite[Theorem 3.1]{khalil_nonlinear_2002} that \cref{eq:Augmented_Dynamics,eq:Augmented_Dynamics_1,eq:Augmented_Dynamics_2} with $\zeta=\zeta_*$ has a unique solution on $[0,t_\rmm)$.

To prove (\ref{thm:safety_b}), let $\tilde x_0 \in \ubar \SC$, and let $\eta \colon [0,t_\rmm) \to \BBR$ satisfy $\dot \eta(t) = -\alpha(\eta(t)) - \beta(t) \eta(t)^2$, where $\eta(0) = h(\tilde x_0)$ and $\beta(t) \triangleq \gamma^{-1}\lambda(t,\tilde x(t))$. 
Since $\alpha$ is nondecreasing, $\alpha(0)=0$, and $\eta(0) =h(\tilde x_0) \ge 0$, it follows that for all $t \in [0,t_\rmm)$, $\eta(t) \ge 0$. 
Moreover, it follows from \cref{eq:Augmented_Dynamics,eq:safe_control,eq:safe_kappa,eq:b,eq:b_ubar}, \cref{prop:ubar_b_lower_b}, \cref{prop:Dist_est_conv} and (\ref{prop:safe_b}) of \cref{prop:safe_control_properties} that for all $t \in [0,t_\rmm)$,
\begin{align*}
    \frac{\rm d}{\rm d t}h(\tilde x(t)) &= h^\prime(\tilde x(t)) \left [ \tilde{f}(\tilde{x}(t))+\tilde{B} \bar \zeta_*(t,\tilde x(t)) + \tilde{D}w(x(t)) \right ]\\
    &= b\big (\tilde x(t), \zeta_*(t,\tilde x(t)),\kappa_*(t,\tilde x(t))\big) 
    - \alpha( h ( \tilde{x}(t) ) ) 
    - \kappa_* \big (t,\tilde x(t) \big ) h(\tilde{x}(t))\\
    &\ge \ubar b \big (\tilde x(t), \zeta_*(t,\tilde x(t)),\kappa_* (t,\tilde x(t)) )- \alpha(h(\tilde{x}(t))) 
    - \kappa_*\big (t,\tilde x(t) \big) h(\tilde{x}(t))\\
    &\ge - \alpha(h(\tilde{x}(t)))
    - \kappa_*(t,\tilde x(t)) h(\tilde{x}(t))\\
    &=- \alpha(h(\tilde{x}(t))) 
    - \beta(t) h(\tilde{x}(t))^2.
\end{align*}
Hence, the comparison lemma \cite[Lemma~3.4]{khalil_nonlinear_2002} implies that for all $t \in [0,t_\rmm)$, $h(\tilde x(t)) \ge \eta(t) \ge 0$.
Thus, for all $t \in [0,t_\rmm)$, $\tilde x(t) \in \SH \subset \ubar \SP$, and \Cref{prop.forward_invariance} implies that for all $t \in [0,t_\rmm)$, $\tilde x(t) \in \ubar \SC$.

To prove (\ref{thm:safety_c}), since for all $t \in [0,t_\rmm)$, $\tilde x(t) \in \ubar \SC$, it follows from \cref{eq:p_bar_set,eq:C_set,eq:C_bar_set} that for all $t \in [0,t_\rmm)$, $x(t) \in \Psi$ and $\mu(t) \in \Phi$.

To prove (\ref{thm:safety_d}), since for all $t \in [0,t_\rmm)$, $x(t) \in \Psi$, it follows that for all $t \in [0,t_\rmm)$, $x(t) \in \Psi_3$ and $x(t) \in \Psi_4$. Then, it follows from \cref{eq:psi_3,eq:psi_4} that $v(t)$ and $\omega(t)$ are bounded.
Moreover, since $R(t) \in \SOT$, it follows that $R(t)$ is bounded.
Furthermore, it follows from \cref{prop:Dist_est_conv} that $\nu(t)$ and $\delta(t)$ are bounded.
\hfill $\blacksquare$
\end{proof}

\Cref{thm:safety} shows that for all $t \in [0,t_\rmm)$, $x(t) \in \Psi$ and $\mu(t) \in \Phi$, thus, achieving \ref{Obj:safety} while producing a control $\mu$ that approximates $\mu_{\rm d}$, which is designed to achieve \ref{Obj:landing}. 
Moreover, the optimal allocation \cref{eq:opt_allocation} can be used to compute an actuator input such that for all $t \in [0,t_\rmm)$, $u_*(\mu(t)) \in \SU$, thus, achieving \ref{Obj:actuator_input}.
The complete feedback control architecture is shown in \cref{fig:ControlArchitecture}.

\begin{figure}[h]
  \centering
  \includegraphics[width=\textwidth]{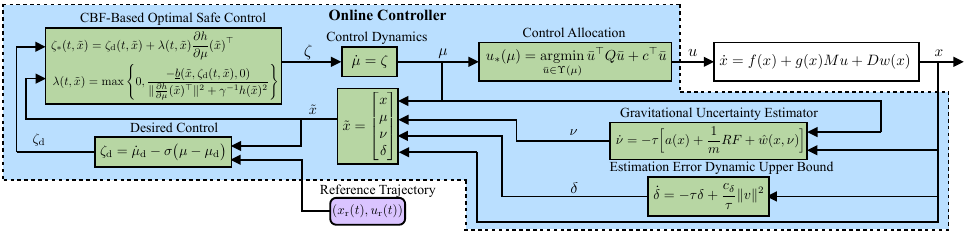}
  \caption{Optimal and safe control with state and input constraints under gravitational uncertainty. Control minimizes cost subject to CBF constraint using the composite CBF $h$.}
  \label{fig:ControlArchitecture}
\end{figure}

\section{Simulations} \label{sec:Simulations}

We consider 2 spacecraft configurations: a robotic lander \cite{ho_mascot_2021} and a crewed landing module \cite{nasa1969lunarmodule}. 
The landing tolerances are $r_{\rm a} = 2.0$~m, $v_{\rm a} = 0.5\,{\text{m}}/{\text{s}}$, $\theta_{\rm a} = 0.2\,\mathrm{rad}$, and $\omega_{\rm a} = 0.05\,{\text{rad}}/{\text{s}}$. 
For the constraints, we let $\theta_{\rm gs} = \pi/3$, $\theta_{\rm b} = \pi/3$, $v_{\rm m} = 12\,{\text{m}}/{\text{s}}$, and $\omega_{\rm m} = 0.1\,{\text{rad}}/{\text{s}}$, and $\varphi$ is given by \cref{example:transition_func} with $k=10^{-5}$.

The unknown gravitation force $G$ is a spherical harmonics model, whereas the known model $G_\rmm$ is a Newtonian point-mass model with 10\% less mass that the true value. 
The spherical harmonics model is implemented using \cite[]{arenas-uribe_higher-order_2026}.
Reference trajectories $(x_{\rm r}, u_{\rm r})$ are obtained with the trajectory generation algorithm \cite[]{malyuta_advances_2021}.

The EHGO \cref{eq:EHGO_Dynamics} and dynamic upper bound \cref{eq:est_error_var_dyn} are implemented with $\tau = 10$, $c_{\delta} = 10^{-3}$, $\nu_0=-\tau v_0$, and $\delta_0 = 10^{-10}$. 
We implement the control \cref{eq:opt_allocation,eq:Fd,eq:Td,eq:ud,eq:control_dynamics,eq:zeta_d,eq:safe_control,eq:lambda_safety} with $\sigma = 5$, $\mu_0 = \mu_{\rm r}(0)$, $\rho = 40$, $\alpha(h) = h$, $\gamma = 10^{20}$, $\beta_{1}(\psi_{1}) = 0.1 \psi_{1}$, $\gamma_{1}(\psi_{1,1}) = 0.1 \psi_{1,1}$, $\beta_{2}(\psi_{2}) = \psi_{2}$, $\gamma_{2}(\psi_{2,1}) = \psi_{2,1}$, $\beta_{3}(\psi_{3}) = \psi_{3}$ and $\beta_{4}(\psi_{4}) = 10 \psi_{4}$. 
The control is implemented with an update rate of 25~Hz. 
The project repository is at \footnote{\url{https://github.com/FelipeArenasUribe/Asteroid_Safe_Landing}}.

\begin{example}\rm \label{ex:robot_lander_ellip}
We consider a robotic lander similar to the Hayabusa2 mission \cite[]{ho_mascot_2021} with mass of $m = 600$~kg and inertia
    \begin{equation*}
        J =\begin{bsmallmatrix} 180 & -5 & -2 \\ -5 & 180 & -8 \\ -2 & -8 & 150 \end{bsmallmatrix} \text{ kg} \cdot \mathrm{m}^2 .
    \end{equation*}
The spacecraft has 4 reaction wheels in a pyramidal arrangement, and 3 bi-directional thrusters in an orthogonal configuration, which yields the allocation matrix
    \begin{equation*}
        M = \frac{1}{3} \begin{bsmallmatrix}
         0 & 0 & 0 & 0 & 3 & 0 & 0 \\
         0 & 0 & 0 & 0 & 0 & 3 & 0 \\
         0 & 0 & 0 & 0 & 0 & 0 & 3 \\
         3 & -1 & -1 & -1 & 0 & 0 & 0 \\
        -1 &  3 & -1 & -1 & 0 & 0 & 0 \\
        -1 & -1 &  3 & -1 & 0 & 0 & 0
        \end{bsmallmatrix}.
    \end{equation*}
Each reaction wheel has a torque limit of $\pm 0.2$~N-m, and each thruster has a force limit of $\pm 20$~N.
Thus, the admissible actuator input set $\SU$ is a hyper-cube with 14 half-spaces, and the resulting admissible control-force-and-torque set $\Phi$ has 18 half-spaces. Control gains are $K_r = 0.1 \, I_3$, $K_v= 0.632 \, I_3$, $K_R = \text{diag}(400, 400, 335)$ and $K_\omega = \text{diag}(536, 536, 448)$.

\begin{subfigure}[t]
    \centering
    \begin{minipage}[b]{0.48\textwidth} 
        \centering
        \includegraphics[width=0.85\linewidth]{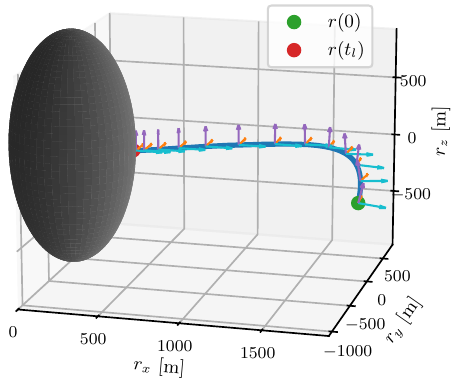}
        \caption{Spacecraft trajectory and body-frame attitude.}
        \label{fig:conf1_ellip_traj}
    \end{minipage}
    \hspace{0.01\textwidth}
    \begin{minipage}[b]{0.48\textwidth}
        \centering
        \includegraphics[width=0.9\linewidth]{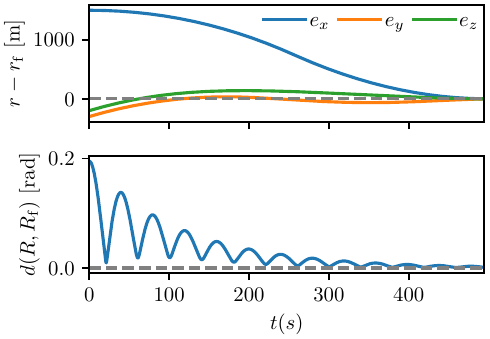}
        \caption{Distance $r-r_\rmf$ to landing, and geodesic distance $d(R,R_\rmf)$ to landing attitude. Note $e_x$, $e_y$, $e_z$ are the entries of $r-r_\rmf$.}
        \label{fig:conf1_ellip_error}
    \end{minipage}
    \hfill
    \setcounter{subfigure}{-1}
    \caption{Robotic lander soft-landing maneuver on elliptical asteroid.}
    \label{fig:ellip_conf1_error_and_traj}
\end{subfigure}
    
The asteroid is ellipsoidal with semi-axes lengths of $1000$~m, $400$~m, and $400$~m, and uniform mass distribution with density $1380$~kg/m$^3$, which is typical for a C-type asteroid. 
The asteroid rotates about its principal axis associated with the largest moment of inertia with synodic period $4$~h.

The initial conditions are 
$r_0 = \begin{bsmallmatrix} 1900 & -300 & -200 \end{bsmallmatrix}^{\top}$~m, 
$v_0 = \begin{bsmallmatrix} 0 & 4 & 4 \end{bsmallmatrix}^{\top}$~m/s, 
$\omega_0 = 0$~rad/s,
and $R_0$ is the rotation corresponding to a 1-2-3 Euler sequence with angles $(0.174,\,0.087,\,0)$~rad. 
The final position is $r_{\rm f} = \begin{bsmallmatrix} 400 & 0 & 0 \end{bsmallmatrix}^{\top}$~m, and the final attitude $R_{\rm f}$ is the rotation corresponding to a 1-2-3 Euler sequence with angles $(0,\,\pi/2,\,0)$~rad.

Landing is achieved at $t_{\rm l} = 494.4$~s, where $\|r-r_{\rm f}\| = 0.786$~m, $\|v - v_{\rm f}\| = 0.147$~m/s, $d(R, R_{\rm f}) = 0.001$~rad, $\|\omega-\omega_{\rm f}\| = 0.001$~rad/s.
Figure \ref{fig:conf1_ellip_traj} shows the spacecraft trajectory, and Figure \ref{fig:conf1_ellip_error} shows the distance $r-r_\rmf$ and geodesic $d(R,R_\rmf)$ throughout the maneuver. 
\Cref{fig:conf1_ellip_h_lambda} shows that the composite CBF $h$ is nonnegative throughout the maneuver, which implies that $\psi_1, \psi_2, \psi_3, \psi_4, \phi_1, \ldots, \phi_{18}$ are nonnegative.
Thus, $x(t) \in \Psi$ and $u(t) \in \SU$ for $t \in [0,t_\rml]$.
\Cref{fig:conf1_ellip_h_lambda} also shows that $\lambda$ is positive for $t \in [467,493]$, which indicates that the desired control forces and torques are modified during this time interval to satisfy state and actuator constraints. 
\Cref{fig:conf1_ellip_state_cons,fig:conf1_ellip_actuator_inputs} confirm that the state and actuator constraints are satisfied throughout the maneuver. 
\Cref{fig:conf1_ellip_dist_est} shows that the magnitude of the estimation error $\tilde w$ and satisfies the upper bound $\sqrt{\delta}$. 
The magnitude of $\tilde w$ increases as the spacecraft approaches the asteroid because the influence of higher-order gravitational effects increases.

\setcounter{figure}{7}
\begin{subfigure}[h]
    \centering
    \begin{minipage}[b]{0.45\textwidth}
        \centering
        \includegraphics[width=0.99\linewidth]{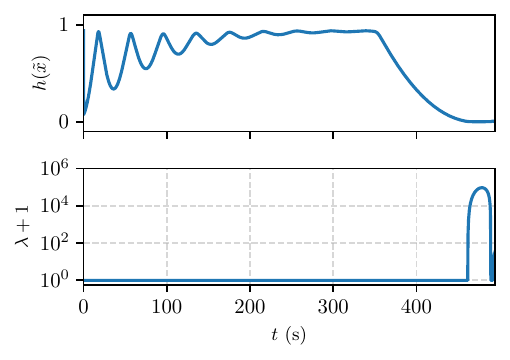}
        \caption{CBF $h$ is nonnegative throughout the maneuver, and $\lambda$ is positive at times when the desired control forces and torques are modified.}
        \label{fig:conf1_ellip_h_lambda}
        \vspace{0.15in} 
        \includegraphics[width=0.99\linewidth]{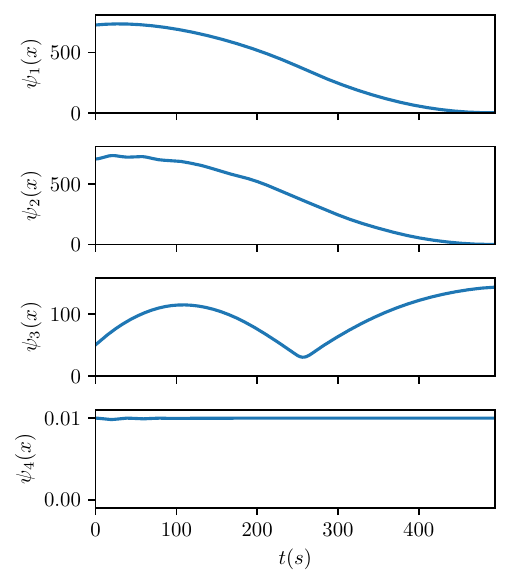}
        \caption{Barrier functions are nonnegative.}
        \label{fig:conf1_ellip_state_cons}
    \end{minipage}
    \hspace{0.1in}
    %
    \begin{minipage}[b]{0.45\textwidth}
        \centering
        \includegraphics[width=0.93\linewidth]{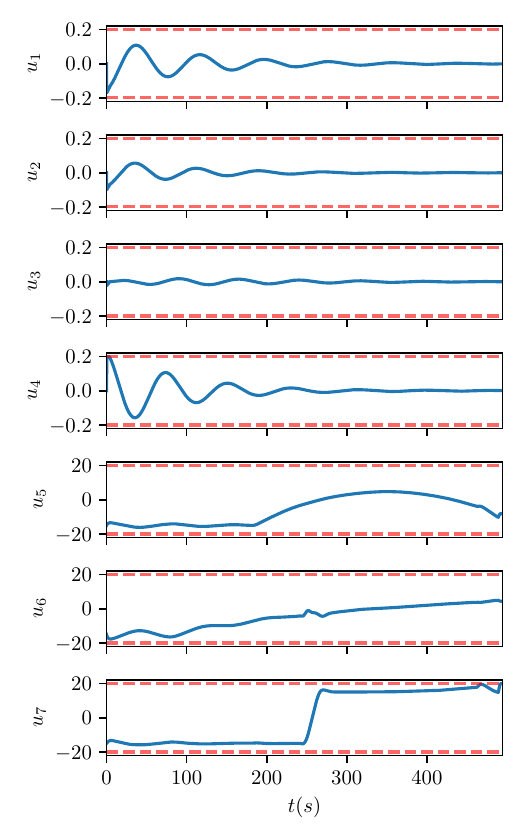}
        \caption{Actuator inputs are in the admissible set (dashed lines).}
        \label{fig:conf1_ellip_actuator_inputs}
        \vspace{0.15in} 
        \includegraphics[width=0.9\linewidth]{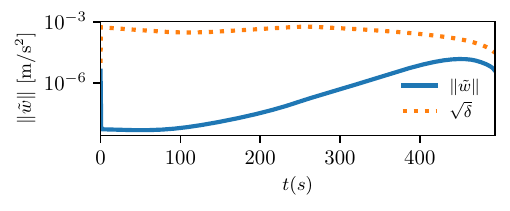}
        \caption{Estimate error $\| \tilde w \|$ is less than the computable upper bound $\sqrt{\delta}$.}
        \label{fig:conf1_ellip_dist_est}
    \end{minipage}

    \setcounter{subfigure}{-1}
    \caption{CBF behavior for the robotic lander and an ellipsoidal asteroid.}
    \label{fig:ellip_conf1_combined}
\end{subfigure}

\end{example}

\begin{example}\rm\label{ex:robotic_lander_irreg}
We again consider the robotic lander from \cref{ex:robot_lander_ellip}, but the asteroid is irregularly shaped with nonuniform density, which is modeled as the union of 2 uniform-density ellipsoids as shown in \Cref{fig:conf1_irreg_traj}. 
One ellipsoid has semi-axes lengths of $1000$~m, $400$~m, and $400$~m, and density $\rho_1 = 1460$~kg/m$^3$
The other has semi-axes lengths of $800$~m, $400$~m, and $400$~m, and density $\rho_2 = 1380$~kg/m$^3$. 
The asteroid rotates about its principal axis associated with the largest moment of inertia, with synodic period $4$~h.

The initial conditions are 
$r_0 = \begin{bsmallmatrix} 1554.259 & -550.0 & -710.605 \end{bsmallmatrix}^{\top}$~m, 
$v_0 = \begin{bsmallmatrix} -1 & 4 & -2 \end{bsmallmatrix}^{\top}$~m/s, 
$\omega_0 = 0$~rad/s, and 
$R_0$ is the rotation corresponding to a 1-2-3 Euler sequence with angles $(0.174,\,0.087,\,0)$~rad. 
The final position is $r_{\rm f} = \begin{bsmallmatrix} 762.523 & 0 & 430.952 \end{bsmallmatrix}^{\top}$~m, and the final attitude $R_{\rm f}$ is the rotation corresponding to a 1-2-3 Euler sequence with angles $(\pi,\,1.13,\,-\pi/2)$~rad.

Landing is achieved at $t_{\rm l} = 495.4 \,\text{ s}$ with $\|r-r_{\rm f}\| = 1.582$~m, $\|v - v_{\rm f}\| = 0.148 $~m/s, $d(R, R_{\rm f}) = 0.001$~rad, and $\|\omega-\omega_{\rm f}\| = 0.002$~rad/s. 
\Cref{fig:conf1_irreg_traj} shows the spacecraft trajectory, and \Cref{fig:conf1_irreg_error} shows $r-r_\rmf$ and $d(R,R_\rmf)$.
\Cref{fig:conf1_irreg_h_lambda} shows that $h$ is nonnegative throughout the maneuver, which implies that $\psi_1, \psi_2, \psi_3, \psi_4, \phi_1, \ldots, \phi_{18}$ are nonnegative.
Thus, $x(t) \in \Psi$ and $u(t) \in \SU$ for $t \in [0,t_\rml]$. 
\Cref{fig:conf1_irreg_h_lambda} also shows that $\lambda$ is positive for $t \in [0,98]$ and $t \in [456,495.4]$. 
In comparison to \cref{fig:conf1_ellip_h_lambda}, there are 2 periods of modification to the desired control forces and torques because of the gravitational effects not captured by the known model $G_{\rm m}$. 
\Cref{fig:conf1_irreg_state_cons,fig:conf1_irreg_actuator_inputs} confirm that the state and actuator constraints are satisfied, and \Cref{fig:conf1_irreg_dist_est} shows that the magnitude of the estimation error $\tilde w$ and satisfies the upper bound $\sqrt{\delta}$.

    \setcounter{figure}{8}
    \begin{subfigure}[h]
        \centering
        \begin{minipage}[b]{0.48\textwidth} 
            \centering
            \includegraphics[width=0.8\linewidth]{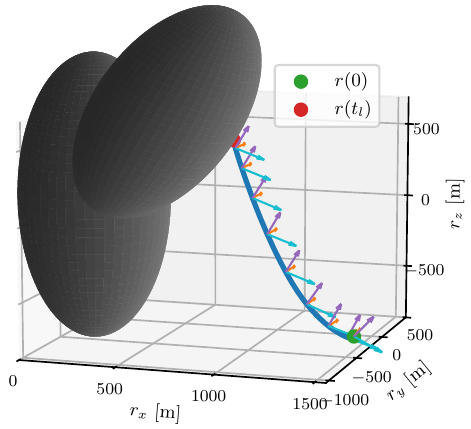}
            \caption{Spacecraft trajectory and body-frame attitude.}
            \label{fig:conf1_irreg_traj}
        \end{minipage}
        \hspace{0.01\textwidth}
        \begin{minipage}[b]{0.48\textwidth}
            \centering
            \includegraphics[width=0.9\linewidth]{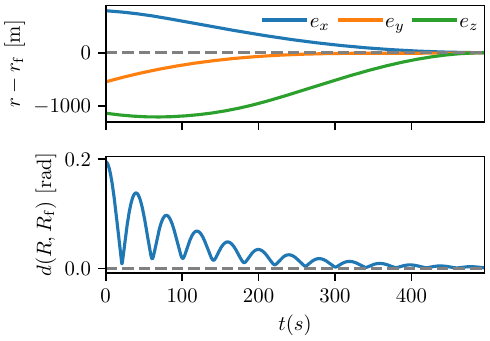}
            \caption{Distance $r-r_\rmf$ to landing, and geodesic distance $d(R,R_\rmf)$ to landing attitude. Note $e_x$, $e_y$, $e_z$ are the entries of $r-r_\rmf$.}
            \label{fig:conf1_irreg_error}
        \end{minipage}
        \hfill
        
        \setcounter{subfigure}{-1}
        \caption{Robotic lander soft-landing maneuver on irregularly-shaped asteroid.}
        \label{fig:irreg_conf1_error_and_traj}
    \end{subfigure}

    \setcounter{figure}{9}
    \begin{subfigure}[t]
        \centering
        \begin{minipage}[b]{0.48\textwidth}
            \centering
            \includegraphics[width=0.99\linewidth]{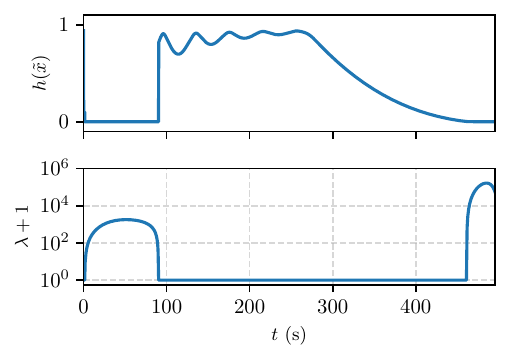}
            \caption{CBF $h$ is nonnegative throughout the maneuver, and $\lambda$ is positive at times when the desired control is modified.}
            \label{fig:conf1_irreg_h_lambda}
            
            \vspace{0.05in} 
            
            \includegraphics[width=0.99\linewidth]{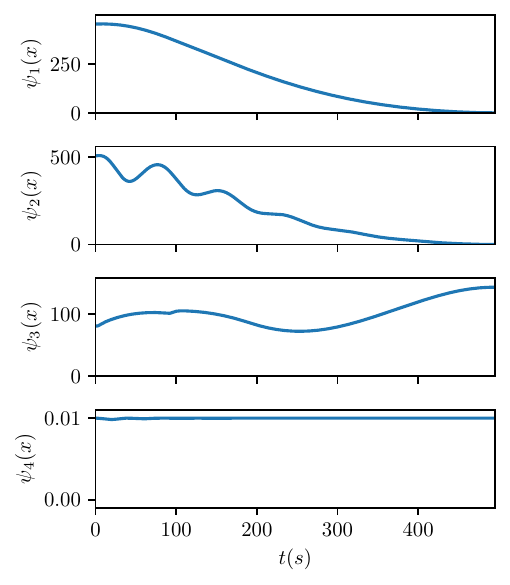}
            \caption{Barrier functions are nonnegative.}
            \label{fig:conf1_irreg_state_cons}
        \end{minipage}
        \hspace{0.01in}
        %
        \begin{minipage}[b]{0.48\textwidth}
            \centering
            \includegraphics[width=0.91\linewidth]{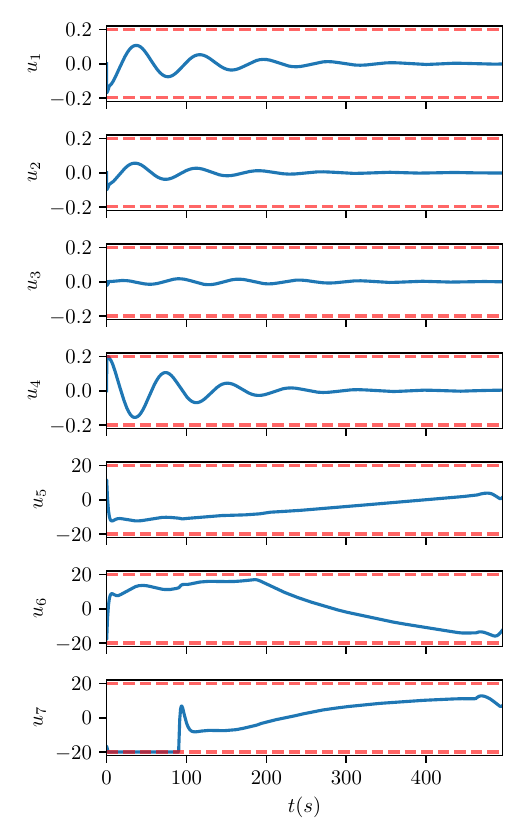}
            \caption{Actuator inputs are in the admissible set (dashed lines).}
            \label{fig:conf1_irreg_actuator_inputs}
            
            \vspace{0.05in} 
            
            \includegraphics[width=0.91\linewidth]{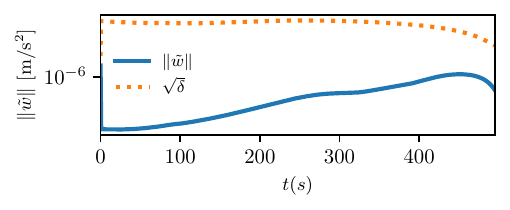}
            \caption{Estimate error $\| \tilde w \|$ is less than the computable upper bound $\sqrt{\delta}$.}
            \label{fig:conf1_irreg_dist_est}
        \end{minipage}

        \setcounter{subfigure}{-1}
        \caption{CBF behavior for the robotic lander and an irregularly-shaped asteroid.}
        \label{fig:irreg_conf1_combined}
    \end{subfigure}
    
\end{example}

\begin{example} \rm
We consider a crewed landing module similar to the Apollo Landing Module \cite[]{nasa1969lunarmodule} with mass of of $m = 3500$~kg and inertia matrix
    \begin{equation*}
        J =\begin{bsmallmatrix} 4200 & 80 & -40 \\ 80 & 4000 & -60 \\ -40 & -60 & 3200 \end{bsmallmatrix} \text{ kg} \cdot \mathrm{m}^2 .
    \end{equation*}

The spacecraft has 8 bi-directional thrusters, which are arranged according to \cite[Section 3-5]{nasa1969lunarmodule} and has the allocation matrix
    \begin{equation*}
        M = \begin{bsmallmatrix}
    1 & 0 & -1 & 0 & 0 & 0 & 0 & 0 \\
    0 & 1 & 0 & -1 & 0 & 0 & 0 & 0 \\
    0 & 0 & 0 & 0 & 1 & 1 & 1 & 1 \\
    0 & 0 & 0 & 0 & 1 & 0 & -1 & 0 \\
    0 & 0 & 0 & 0 & 0 & 1 & 0 & -1 \\
    1 & 1 & 1 & 1 & 0 & 0 & 0 & 0
    \end{bsmallmatrix}.
    \end{equation*}
Each thruster has a force limit of $\pm 100$~N.
Thus, the admissible actuator input set $\SU$ is a hyper-cube with 16 half-spaces, and the resulting admissible control-force-and-torque set $\Phi$ has 24 half-spaces. 
Control gains are $K_r = 0.1 \, I_3$, $K_v= 0.632 \, I_3$, $K_R = \text{diag}(189, 180, 144) \cdot 10^4$, and $K_\omega = \text{diag}(63, 59.4, 46.8) \cdot 10^4$. 
We use the asteroid, initial conditions, and final conditions from \cref{ex:robotic_lander_irreg}.

Landing is achieved at $t_{\rm l} = 469.4$~s, where $\|r-r_{\rm f}\| = 0.786$~m, $\|v - v_{\rm f}\| = 0.147$~m/s, $d(R, R_{\rm f}) = 0.001$~rad, $\|\omega-\omega_{\rm f}\| = 0.001$~rad/s.
\Cref{fig:conf2_irreg_traj} shows the spacecraft trajectory, and \Cref{fig:conf2_irreg_error} shows $r-r_\rmf$ and $d(R,R_\rmf)$.
\Cref{fig:conf2_irreg_h_lambda} shows that $h$ is nonnegative throughout the maneuver, which implies that $\psi_1, \psi_2, \psi_3, \psi_4, \phi_1, \ldots, \phi_{24}$ are nonnegative.
Thus, $x(t) \in \Psi$ and $u(t) \in \SU$ for $t \in [0,t_\rml]$. 
\Cref{fig:conf2_irreg_h_lambda} also shows that $\lambda$ is positive for $t \in [450,469.4]$. \Cref{fig:conf2_irreg_state_cons,fig:conf2_irreg_actuator_inputs} confirm that the state and actuator constraints are satisfied throughout the maneuver. 
The gravity estimator has similar performance to  \Cref{fig:conf1_irreg_dist_est}.

    \setcounter{figure}{10}
    \begin{subfigure}[t]
        \centering
        \begin{minipage}[b]{0.48\textwidth} 
            \centering
            \includegraphics[width=0.85\linewidth]{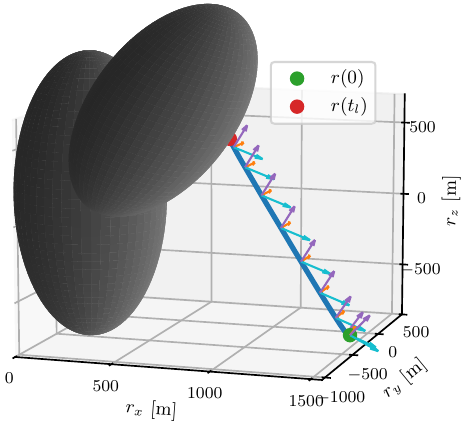}
            \caption{Spacecraft trajectory and body-frame attitude.}
            \label{fig:conf2_irreg_traj}
        \end{minipage}
        \hspace{0.01\textwidth}
        \begin{minipage}[b]{0.48\textwidth}
            \centering
            \includegraphics[width=0.9\linewidth]{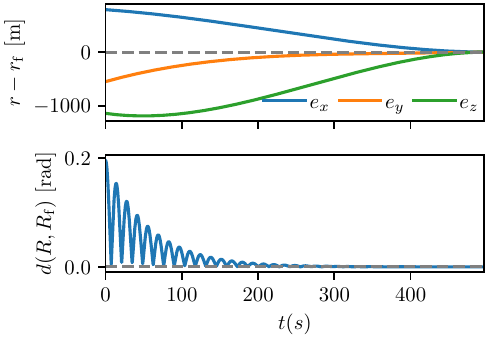}
            \caption{Distance $r-r_\rmf$ to landing, and geodesic distance $d(R,R_\rmf)$ to landing attitude. Note $e_x$, $e_y$, $e_z$ are the entries of $r-r_\rmf$.}
            \label{fig:conf2_irreg_error}
        \end{minipage}
        \hfill
        
        \setcounter{subfigure}{-1}
        \caption{Crewed landing module soft-landing maneuver on irregularly-shaped asteroid.}
        \label{fig:irreg_conf2_error_and_traj}
    \end{subfigure}

    \setcounter{figure}{11}
    \begin{subfigure}[ht]
        \centering
        \begin{minipage}[b]{0.48\textwidth}
            \centering
            \includegraphics[width=0.91\linewidth]{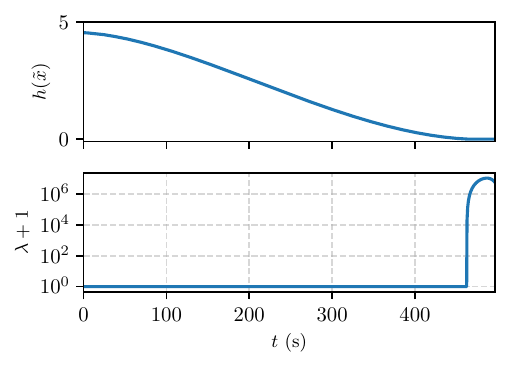}
            \caption{CBF $h$ is nonnegative throughout the maneuver, and $\lambda$ is positive at times when the desired control is modified.}
            \label{fig:conf2_irreg_h_lambda}
            
            \vspace{0.05in} 
            
            \includegraphics[width=0.91\linewidth]{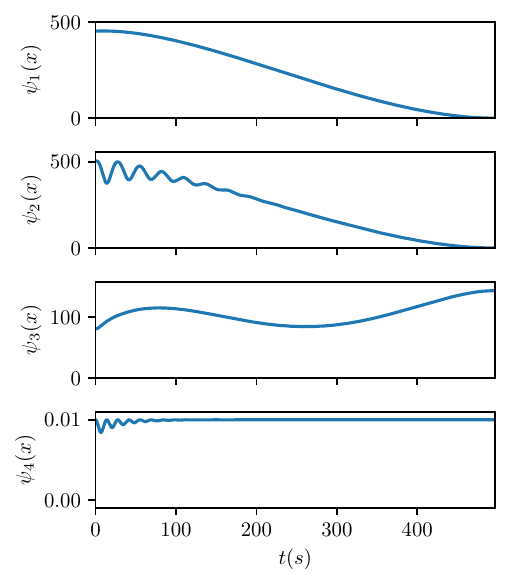}
            \caption{State barrier functions are nonnegative.} 
            \label{fig:conf2_irreg_state_cons}
        \end{minipage}
        \hspace{0.01in}
        %
        \begin{minipage}[b]{0.48\textwidth}
            \centering
            \includegraphics[width=0.975\linewidth]{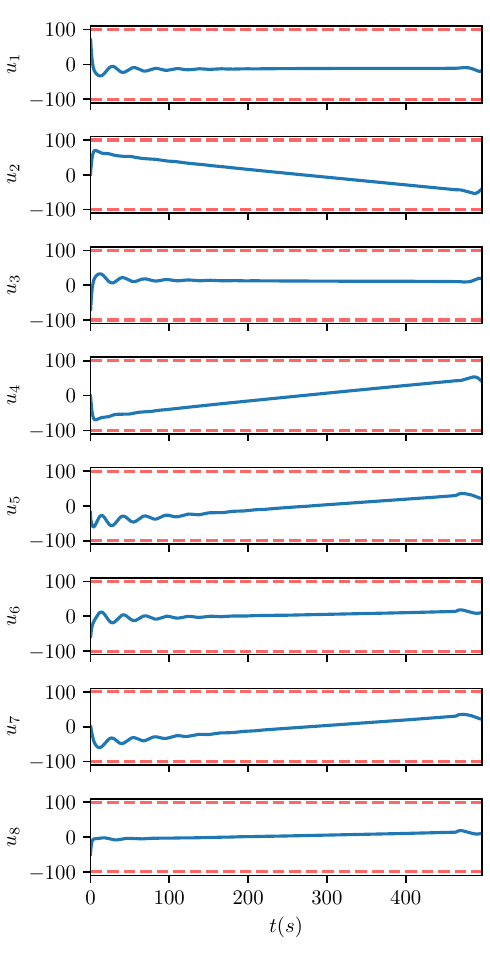}
            \caption{Actuator inputs are in the admissible set (dashed lines).}
            \label{fig:conf2_irreg_actuator_inputs}

        \end{minipage}

        \setcounter{subfigure}{-1}
        \caption{CBF behavior for the crewed landing module and an irregularly-shaped asteroid.}
        \label{fig:irreg_conf2_combined}
    \end{subfigure}
\end{example}

\section{Conclusions}\label{sec:Conclusions}

This article presents a control architecture for safe soft-landing maneuvers of overactuated and fully actuated spacecraft on small celestial bodies under gravitational uncertainty. Several key contributions enable simultaneous position and attitude control while satisfying state and input constraints.

The constraint projection method presented in \cref{sec:Constraint_Projection} is a key enabler to re-frame the actuator constraint $u \in \SU$ as a constraint on the forces and torques $\mu \in \Phi$. Thus, allowing to focus on designing control forces and torques instead of individual actuator commands. Moreover, the allocation presented in \cref{eq:opt_allocation} provides a mapping from the forces and torques to the actuator inputs given by a quadratic program which \cref{prop:feasibility} proves to be feasible and has a unique solution. It is worth noting that other allocation algorithms may be used to perform such mapping provided that they enforce the constraint $u \in \SU$. 

The gravitational uncertainty estimation $\hat w$ presented in \cref{sec:FeedbackControl} is an essential element of the control architecture. It is not only used to compensate for the disturbance in \cref{eq:Fd}. But also, by leveraging on $\delta$ and \Cref{prop:Dist_est_conv}, we design a computable lower bound of the CBF constraint function $\cref{eq:b}$ as shown in \Cref{prop:ubar_b_lower_b}. Finally in \cref{sec:CBFControl}, we present a minimum-intervention control which is shown to be continuous in \cref{prop:safe_control_properties} and \cref{thm:safety} shows that it satisfies \ref{Obj:safety} and \ref{Obj:actuator_input} while producing a control that approximates the desired control \cref{eq:ud}, which is designed to achieve \ref{Obj:landing}.

The simulation results demonstrate that the proposed method enables successful landings under significant gravitational uncertainty where traditional feedback control fails. Furthermore, the closed-form nature of the controller ensures the computational efficiency required for autonomous onboard implementation.

A limitation of the polytope projection method is the growth in the number of half-spaces $\ell$, which depends on the geometry of $\SU$ and the allocation matrix $M$. We note that the projection of any convex set through the affine mapping \eqref{eq:allocation} is a convex set, indicating that the input constraint projection method can be extended to more general convex sets. However, finding efficient representations for $\Phi$ when $\SU$ is a non-polytopic convex set is open for future work.

\newpage

\appsection{Proofs of Proposition \ref{prop:Dist_est_conv} and \ref{lemma:position_control_stability}}\label{app:Proofs}

\begin{proof}[Proof of \Cref{prop:Dist_est_conv}]
Since $\tilde w(x,\nu) = w(x) - \tau v -\nu$, it follows that the time derivative of \cref{eq:xi} along \cref{eq:system_dynamics,eq:system_dynamics.2,eq:system_dynamics.3,eq:system_dynamics.4,eq:EHGO_Dynamics,eq:est_error_var_dyn} is
    \begin{align}
        \dot \xi &= \dot \delta - 2 \tilde w^\top ( \dot w - \tau \dot v - \dot \nu ) \nn\\
        & = -\tau \delta + \frac{c_\delta}{\tau}\|v\|^2 - 2 \tilde w^\top \dot w + 2 \tau \|\tilde w\|^2. \label{eq:xi_dot}
    \end{align}
    %
Since $0 \le  \| \sqrt{\tau} \tilde w - \dot w / \sqrt{\tau} \|^2 = \tau \| \tilde w \|^2 - 2\tilde w^\top \dot w + \| \dot w \|^2 /\tau$, it follows that 
\begin{equation*}
    \tau \| \tilde w \|^2 - 2 \tilde w^\rmT \dot w  \ge - \frac{\| \dot w \|^2}{ \tau},
\end{equation*}
which combined with \cref{eq:xi_dot} implies that 
 \begin{equation}
\dot \xi \ge -\tau \delta + \frac{c_\delta}{\tau}\|v\|^2 - \frac{\| \dot w \|^2}{ \tau} +\tau \| \tilde w \|^2. \label{eq:xi_dot.2}
\end{equation}
Next, since $w(x) = \Delta(r)$, it follows that $\dot w = \Delta^\prime(r) v$.
Moreover, it follows from Assumption~\ref{ass:Bounded_Dist} that there exists $c_\delta > 0$ such that $c_\delta \ge \| \Delta^\prime(r) \|^2$.
Thus, it follows that $\| \dot w \|^2 \le \| \Delta^\prime(r)\|^2 \|v\|^2 \le c_\delta \|v\|^2$, which combined with \cref{eq:xi_dot.2} yields 
 \begin{equation}
\dot \xi \ge -\tau \delta + \tau \|\tilde w\|^2 = -\tau \xi. \label{eq:xi_dot.3}
\end{equation}
Since $\xi(x_0, \nu_0, \delta_0) \geq 0$, it follows from \cref{eq:xi_dot.3} that for all $t\geq 0$, $\xi(x(t), \nu(t), \delta(t)) \geq 0$.
\hfill $\blacksquare$
\end{proof}
\begin{proof}[Proof of \Cref{lemma:position_control_stability}]
It follows from \cref{eq:system_dynamics,eq:system_dynamics.2,eq:system_dynamics.3,eq:system_dynamics.4,eq:EHGO_Dynamics,eq:est_error_var_dyn} that
    \begin{equation}\label{eq:error_system}
        \dot{z}(t) = A z(t) + B\tilde{w}(x(t), \nu(t)),
    \end{equation}
    where
    \begin{equation*}
        z(t) \triangleq \begin{bmatrix} e_{r}(t,x(t)) \\ e_v(t,x(t)) \end{bmatrix}, \qquad A = \begin{bmatrix} 0_{3 \times 3} & I_{3 \times 3} \\ -K_r & -K_v \end{bmatrix}, \qquad B = \begin{bmatrix} 0_{3 \times 3} \\ I_{3 \times 3} \end{bmatrix}.
    \end{equation*}
Since $K_r$ and $K_v$ are positive definite, it follows that $A$ is asymptotically stable. 
Let $P \in \BBR^{6 \times 6}$ be the positive-definite solution to $A^\top P + PA = - I$, and define $V(t) \triangleq z(t)^\top Pz(t)$. 
The time derivative of $V$ along the trajectories of \cref{eq:error_system} is
    \begin{equation}\label{eq:V_dot}
        \dot{V} = -z^\top z + 2z^\top PB\tilde{w}.
    \end{equation}
Next, let $\varsigma \in (0,1)$ be such that $\varsigma^2 \neq 1 - \tau \lambda_{\max}(P)$, and it follows from \cref{eq:V_dot} that
    \begin{align}
        \dot{V} =& -z^\top z -\big(\varsigma z - \frac{1}{\varsigma}PB\tilde{w}\big)^\top \big(\varsigma z - \frac{1}{\varsigma}PB\tilde{w}\big) + \varsigma^2 z^\top z + \frac{1}{\varsigma^2}\tilde{w}^\top B^\top P P B \tilde{w}, \notag \\
        \leq& -(1 - \varsigma^2)z^\top z + \frac{\lambda_{\max}(B^\top PPB)}{\varsigma^2} \tilde{w}^\top \tilde{w}. \label{eq:V_dot_bound}
    \end{align}
Since $V \leq \lambda_{\max}(P)z^\top z$, it follows that $-z^\top z \leq - {V}/{\lambda_{\max}(P)}$, which combined with \cref{eq:V_dot_bound} yields
    \begin{equation}
        \dot{V} \leq -\frac{(1-\varsigma^2)}{\lambda_{\max}(P)}V + \frac{\lambda_{\max}(B^\top PPB)}{\varsigma^2} \tilde{w}^\top \tilde{w}. \label{eq:V_dot_bound_1}
    \end{equation}
Since $\xi(x_0, \nu_0, \delta_0) \geq 0$, \cref{prop:Dist_est_conv,eq:V_dot_bound_1} imply that
    \begin{equation}
        \dot{V} \leq -\frac{(1-\varsigma^2)}{\lambda_{\max}(P)}V + \frac{\lambda_{\max}(B^\top PPB)}{\varsigma^2} \delta. \label{eq:V_dot_bound_2}
    \end{equation}
Since the solution to \cref{eq:est_error_var_dyn} is
    \begin{align*}
        \delta(t) &= \delta_0e^{-\tau t} + \frac{c_{\delta}}{\tau} \int_0^t e^{-\tau (t-s)}\|v(s)\|^2 \, {\rm d}s \notag\\
        &\leq \delta_0e^{-\tau t} + \frac{c_{\delta}}{\tau} \int_0^t e^{-\tau (t-s)}v_{\rm m}^2 \, {\rm d}s, \notag \\
        &= \delta_0e^{-\tau t} + \frac{c_{\delta}v_{\rm m}^2}{\tau^2}(1-e^{-\tau t}), 
    \end{align*}
substituting into \cref{eq:V_dot_bound_2} yields
\begin{equation}\label{eq:V_dot_bound_3}
\dot{V}(t) \leq -c_1V(t) + \frac{c_2}{\tau^2} + c_3 e^{- \tau t},
\end{equation}
where 
\begin{equation*}
c_1 = \frac{(1-\varsigma^2)}{\lambda_{\max}(P)}, \qquad 
c_2 = \frac{\lambda_{\max}(B^\top PPB) c_{\delta} v_{\rm m}^2}{\varsigma^2}, \qquad 
c_3 = (\delta_0 - \frac{c_{\delta} v_{\rm m}^2}{\tau^2}) \frac{\lambda_{\max}(B^\top PPB)}{\varsigma^2}.
\end{equation*}
Next, consider $\bar{V}: [0,\infty) \to \BBR$ that satisfies
    \begin{equation} \label{eq:V_bar_dot}
        \dot{\bar{V}}(t) = -c_1\bar{V}(t)  + \frac{c_2}{\tau^2} + c_3 e^{- \tau t},
    \end{equation}
where $\bar{V}(0) = V(0)$.
Since $\varsigma^2 \neq 1 - \tau \lambda_{\max}(P)$, it follows that $c_1 \neq \tau$, which implies that the solution to \cref{eq:V_bar_dot} is
\begin{equation}\label{eq:Vbar_solution}
\bar{V}(t) = \bar{V}(0)e^{-c_1t} + \frac{c_2}{c_1\tau^2} (1 - e^{-c_1t}) + \frac{c_3}{c_1-\tau} (e^{-\tau t} - e^{-c_1t}).
\end{equation}
Using the comparison lemma \cite[Lemma~3.4]{khalil_nonlinear_2002} on \cref{eq:V_bar_dot,eq:V_dot_bound_3} yields $V(t) \leq \bar{V}(t)$.
Since, in addition, $\lambda_{\min}(P) z^\top z \leq V$, it follows that for all $t \geq 0$, $z(t)^\top z(t) \leq \bar{V}(t)/\lambda_{\min}(P)$.
Finally, since \cref{eq:Vbar_solution} implies that $\lim_{t \to \infty} \bar{V}(t) = \frac{c_2}{c_1\tau^2}$, letting $\vartheta \triangleq \sqrt{\frac{c_2}{c_1 \lambda_{\min}(P)}}$ yields that $\| z(t) \| \to \left [ 0, \frac{\vartheta}{\tau} \right ]$ as $t \to \infty$.
\hfill $\blacksquare$
\end{proof}

\bibliographystyle{Frontiers-Harvard}
\bibliography{references}

\end{document}